\documentclass[letterpaper,onecolumn,11pt,accepted=2019-12-03]{quantumarticle}
\pdfoutput=1
\usepackage[english]{babel}

\usepackage{amssymb,amsmath,amsthm,amsfonts}
\usepackage{mathtools}
\usepackage{enumitem}
\usepackage[numbers,comma,sort&compress]{natbib}
\usepackage{authblk}
\usepackage{graphicx}
\usepackage[font=small]{caption}
\usepackage[labelformat=simple]{subcaption}
\usepackage{float}
\usepackage[ruled,vlined,linesnumbered]{algorithm2e}
\usepackage{physics}
\usepackage{footnote}
\usepackage{xcolor}

\usepackage{tikz}
\usetikzlibrary{calc,patterns,angles,quotes}
\usepackage{siunitx}

\usepackage[colorlinks]{hyperref}

\newtheorem{theorem}{Theorem}
\newtheorem{definition}{Definition}
\newtheorem{lemma}{Lemma}
\newtheorem{proposition}{Proposition}

\numberwithin{equation}{section}
\numberwithin{theorem}{section}
\numberwithin{definition}{section}
\numberwithin{lemma}{section}
\numberwithin{proposition}{section}
\numberwithin{example}{section}
\numberwithin{corollary}{section}
\numberwithin{remark}{section}

\newcommand{\eqn}[1]{(\ref{eqn:#1})}
\newcommand{\eq}[1]{(\ref{eq:#1})}
\newcommand{\prb}[1]{(\ref{prb:#1})}
\newcommand{\thm}[1]{\hyperref[thm:#1]{Theorem~\ref*{thm:#1}}}
\newcommand{\cor}[1]{\hyperref[cor:#1]{Corollary~\ref*{cor:#1}}}
\newcommand{\defn}[1]{\hyperref[defn:#1]{Definition~\ref*{defn:#1}}}
\newcommand{\lem}[1]{\hyperref[lem:#1]{Lemma~\ref*{lem:#1}}}
\newcommand{\prop}[1]{\hyperref[prop:#1]{Proposition~\ref*{prop:#1}}}
\newcommand{\fig}[1]{\hyperref[fig:#1]{Figure~\ref*{fig:#1}}}
\newcommand{\tab}[1]{\hyperref[tab:#1]{Table~\ref*{tab:#1}}}
\newcommand{\algo}[1]{\hyperref[algo:#1]{Algorithm~\ref*{algo:#1}}}
\renewcommand{\sec}[1]{\hyperref[sec:#1]{Section~\ref*{sec:#1}}}
\newcommand{\append}[1]{\hyperref[append:#1]{Appendix~\ref*{append:#1}}}
\newcommand{\fac}[1]{\hyperref[fac:#1]{Fact~\ref*{fac:#1}}}
\newcommand{\lin}[1]{\hyperref[lin:#1]{Line~\ref*{lin:#1}}}
\newcommand{\itm}[1]{\ref{itm:#1}}

\newcommand{\algname}[1]{\textup{\texttt{#1}}}

\newcommand{\hence}{\Rightarrow}
\newcommand{\softoh}[1]{\tilde{O}(#1)}

\newcommand{\grid}[2]{G_{#1 , #2}}

\def\>{\rangle}
\def\<{\langle}

\newcommand{\N}{\mathbb{N}}
\newcommand{\Z}{\mathbb{Z}}
\newcommand{\R}{\mathbb{R}}

\newcommand{\E}{\mathbb{E}}

\DeclareMathOperator{\poly}{poly}
\DeclareMathOperator{\sgn}{sgn}

\DeclareMathOperator{\Ord}{Ord}

\DeclareMathOperator{\EVAL}{EVAL}
\DeclareMathOperator{\MEM}{MEM}
\DeclareMathOperator{\SEP}{SEP}
\DeclareMathOperator{\OPT}{OPT}

\DeclareMathOperator{\QFT}{QFT}

\newcommand{\devlin}{\Delta}

\renewcommand{\d}{\mathrm{d}}
\renewcommand{\emptyset}{\varnothing}
\def\Tr{\operatorname{Tr}}\def\:{\hbox{\bf:}}
\newcommand{\range}[1]{[#1]}

\newcommand{\hd}[1]{\vspace{4mm} \noindent \textbf{#1}}
\def \eps {\epsilon}

\makeatletter
\newcommand{\doubletilde}[1]{{%
  \mathpalette\double@tilde{#1}%
}}
\newcommand{\double@tilde}[2]{%
  \sbox\z@{$\m@th#1\tilde{#2}$}%
  \ht\z@=.9\ht\z@
  \tilde{\box\z@}%
}
\makeatother

\let\oldnl\nl
\newcommand{\nonl}{\renewcommand{\nl}{\let\nl\oldnl}}

\begin{document}

\title{Quantum algorithms and lower bounds for convex optimization}

\author[aff1]{Shouvanik Chakrabarti}
\author[aff1]{Andrew M.\ Childs}
\author[aff1]{Tongyang Li}
\author[aff1]{Xiaodi Wu}
\affiliation[aff1]{Department of Computer Science, Institute for Advanced Computer Studies, and Joint Center for Quantum Information and Computer Science, University of Maryland}

\maketitle

\begin{abstract}
While recent work suggests that quantum computers can speed up the solution of semidefinite programs, little is known about the quantum complexity of more general convex optimization. We present a quantum algorithm that can optimize a convex function over an $n$-dimensional convex body using $\tilde{O}(n)$ queries to oracles that evaluate the objective function and determine membership in the convex body. This represents a quadratic improvement over the best-known classical algorithm. We also study limitations on the power of quantum computers for general convex optimization, showing that it requires $\tilde{\Omega}(\sqrt n)$ evaluation queries and $\Omega(\sqrt{n})$ membership queries.
\end{abstract}


\section{Introduction} \label{sec:intro}

Convex optimization has been a central topic in the study of mathematical optimization, theoretical computer science, and operations research over the last several decades. On the one hand, it has been used to develop numerous algorithmic techniques for problems in combinatorial optimization, machine learning, signal processing, and other areas. On the other hand, it is a major class of optimization problems that admits efficient classical algorithms \cite{boyd2004optimization,grotschel2012geometric}. Approaches to convex optimization include the ellipsoid method \cite{grotschel2012geometric}, interior-point methods \cite{karmarkar1984new,dantzig2006linear}, cutting-plane methods \cite{kelley1960cutting,vaidya1989new}, and random walks \cite{kalai2006simulated,lovasz2006fast}.

The fastest known classical algorithm for general convex optimization solves an instance of dimension $n$ using $\tilde{O}(n^{2})$ queries to oracles for the convex body and the objective function, and runs in time $\tilde{O}(n^{3})$ \cite{lee2017efficient}.\footnote{\label{footnote:tilde}The notation $\tilde{O}$ suppresses poly-logarithmic factors in $n,R,r,\epsilon$, i.e., $\tilde{O}(f(n))=f(n)\log^{O(1)}(nR/r\epsilon)$.} The novel step of \cite{lee2017efficient} is a construction of a separation oracle by a subgradient calculation with $O(n)$ objective function calls and $O(n)$ extra time. It then relies on a reduction from optimization to separation that makes $\tilde{O}(n)$ separation oracle calls and runs in time $\tilde{O}(n^3)$ \cite{lee2015faster}. Although it is unclear whether the query complexity of $\tilde{O}(n^2)$ is optimal for all possible classical algorithms, it is the best possible result using the above framework. This is because it takes $\tilde{\Omega}(n)$ queries to compute the (sub)gradient (see \sec{classical-gradient}) and it also requires $\Omega(n)$ queries to produce an optimization oracle from a separation oracle (see \cite{nemirovsky1983problem} and \cite[Section 10.2.2]{nemirovski1995information}).

It is natural to ask whether quantum computers can solve convex optimization problems faster. Recently, there has been significant progress on quantum algorithms for solving a special class of convex optimization problems called semidefinite programs (SDPs). SDPs generalize the better-known linear programs (LPs) by allowing positive semidefinite matrices as variables. For an SDP with $n$-dimensional, $s$-sparse input matrices and $m$ constraints, the best known classical algorithm \cite{lee2015faster} finds a solution in time $\tilde{O}(m(m^2+n^{\omega}+mns)\poly\log(1/\eps))$, where $\omega$ is the exponent of matrix multiplication and $\eps$ is the accuracy of the solution.
Brand{\~a}o and Svore gave the first quantum algorithm for SDPs with worst-case running time $\tilde{O}(\sqrt{mn} s^2 (Rr/\varepsilon)^{32})$, where $R$ and $r$ upper bound the norms of the optimal primal and dual solutions, respectively \cite{brandao2016quantum}. Compared to the aforementioned classical SDP solver \cite{lee2015faster}, this gives a polynomial speedup in $m$ and $n$. Van Apeldoorn et al.\ \cite{vanApeldoorn2017quantum} further improved the running time of a quantum SDP solver to $\tilde{O}(\sqrt{mn} s^2 (Rr/\eps)^8)$, which was subsequently improved to $\tilde{O}\big((\sqrt{m}+\sqrt{n}(Rr/\eps))s (Rr/\eps)^4\big)$ \cite{brandao2017exponential,vanApeldoorn2018SDP}.
The latter result is tight in the dependence of $m$ and $n$ since there is a quantum lower bound of $\Omega(\sqrt{m}+\sqrt{n})$ for constant $R,r,s,\epsilon$~\cite{brandao2016quantum}.

However, semidefinite programming is a structured form of convex optimization that does not capture the problem in general. In particular, SDPs are specified by positive semidefinite matrices, and their solution is related to well-understood tasks in quantum computation such as solving linear systems (e.g.,~\cite{harrow2009quantum,childs2015quantum}) and Gibbs sampling (e.g.,~\cite{brandao2017exponential,vanApeldoorn2018SDP}). General convex optimization need not include such structural information, instead only offering the promise that the objective function and constraints are convex. Currently, little is known about whether quantum computers could provide speedups for general convex optimization. Our goal is to shed light on this question.

\subsection{Convex optimization}
We consider the following general minimization problem:
  \begin{equation}
    \label{prb:convex-opt}
    \min_{x \in K} f(x),\text{ where $K  \subseteq \mathbb{R}^n$ is a convex set and $f\colon K\rightarrow\R$ is a convex function.}
  \end{equation}
We assume we are given upper and lower bounds on the function values, namely $m \le \min_{x \in K} f(x) \le M$, and inner and outer bounds on the convex set $K$, namely
\begin{align}\label{eqn:ball-relationship}
B_2(0,r) \subseteq K \subseteq B_2(0,R),
\end{align}
where $B_2(x,l)$ is the ball of radius $l$ in $L_{2}$ norm centered at $x\in\R^{n}$. We ask for a solution $\tilde{x} \in K$ with precision $\epsilon$, in the sense that
  \begin{equation}
    \label{eqn:approx}
    f(\tilde x)\leq\min_{x \in K} f(x)+\epsilon.
  \end{equation}

We consider the very general setting where the convex body $K$ and convex function $f$ are only specified by oracles. In particular, we have:
\begin{itemize}
\item A \emph{membership oracle} $O_K$ for $K$, which determines whether a given $x\in\R^{n}$ belongs to $K$;
\item An \emph{evaluation oracle} $O_f$ for $f$, which outputs $f(x)$ for a given $x \in K$.
\end{itemize}

Convex optimization has been well-studied in the model of membership and evaluation oracles since this provides a reasonable level of abstraction of $K$ and $f$, and it helps illuminate the algorithmic relationship between the optimization problem and the relatively simpler task of determining membership \cite{grotschel2012geometric,lee2015faster,lee2017efficient}. The efficiency of convex optimization is then measured by the number of queries to the oracles (i.e., the \emph{query complexity}) and the total number of other elementary gates (i.e., the \emph{gate complexity}).

It is well known that a general bounded convex optimization problem is equivalent to one with a linear objective function over a different bounded convex set. In particular, if promised that $\min_{x \in K} f(x) \le M$, \prb{convex-opt} is equivalent to the problem
\begin{align}\label{prb:gen_opt}
\min_{x' \in \R,\,x \in K}x' \quad\text{such that}\quad f(x)\leq x'\leq M.
\end{align}
Observe that a membership query to the new convex set
\begin{align}\label{eq:kprime}
  K':=\{(x',x) \in \R \times K \mid f(x)\leq x'\leq M\}
\end{align}
can be implemented with one query to the membership oracle for $K$ and one query to the evaluation oracle for $f$. Thus the ability to optimize a linear function
\begin{align}
    \label{eq:rephrased-convex}
    \min_{x \in K} c^Tx
\end{align}
for any $c \in \R^n$ and convex set $K \subseteq \R^n$ is essentially equivalent to solving a general convex optimization problem. A procedure to solve such a problem for any specified $c$ is known as an \emph{optimization oracle}. Thus convex optimization reduces to implementing optimization oracles over general convex sets (\lem{sco}). The related concept of a \emph{separation oracle} takes as input a point $p \notin K$ and outputs a hyperplane separating $p$ from $K$.

In the quantum setting, we model oracles by unitary operators instead of classical procedures. In particular, in the quantum model of membership and evaluation oracles, we are promised to have unitaries $O_{K}$ and $O_{f}$ such that
\begin{itemize}
\item For any $x\in\R^{n}$, $O_{K}|x,0\>=|x,\delta[x \in K]\>$, where $\delta[P]$ is $1$ if $P$ is true and $0$ if $P$ is false;
\item For any $x\in \R^{n}$, $O_{f}|x,0\>=|x,f(x)\>$.
\end{itemize}

In other words, we allow coherent superpositions of queries to both oracles. If the classical oracles can be implemented by explicit circuits, then the corresponding quantum oracles can be implemented by quantum circuits of about the same size, so the quantum query model provides a useful framework for understanding the quantum complexity of convex optimization.

\subsection{Contributions}\label{sec:contributions}

We now describe the main contributions of this paper. Our first main result is a quantum algorithm for optimizing a convex function over a convex body. Specifically, we show the following:

\begin{theorem}\label{thm:upper-main}
There is a quantum algorithm for minimizing a convex function $f$ over a convex set $K \subseteq \R^n$ using $\tilde{O}(n)$ queries to an evaluation oracle for $f$ and $\tilde{O}(n)$ queries to a membership oracle for $K$. The gate complexity of this algorithm is $\tilde{O}(n^3)$.
\end{theorem}

Recall that the state-of-the-art classical algorithm \cite{lee2017efficient} for general convex optimization with evaluation and membership oracles uses $\tilde{O}(n^2)$ queries to each. Thus our algorithm provides a quadratic improvement over the best known classical result. While the query complexity of \cite{lee2017efficient} is not known to be tight, it is the best possible result that can be achieved using subgradient computation to implement a separation oracle, as discussed above.

The proof of \thm{upper-main} follows the aforementioned classical strategy of constructing a separating hyperplane for any given point outside the convex body \cite{lee2017efficient}. We find this hyperplane using a fast quantum algorithm for gradient estimation using $\tilde{O}(1)$ evaluation queries,\footnote{Here $\tilde{O}(1)$ has the same definition as footnote \ref{footnote:tilde}, i.e., $\tilde{O}(1)=\log^{O(1)}(nR/r\epsilon)$.} as first proposed by Jordan \cite{jordan2005fast} and later refined by \cite{gilyen2019optimizing} with more rigorous analysis. However, finding a suitable hyperplane in general requires calculating approximate subgradients of convex functions that may not be differentiable, whereas the algorithms in \cite{jordan2005fast} and \cite{gilyen2019optimizing} both require bounded second derivatives or more stringent conditions. To address this issue, we introduce classical randomness into the algorithm to produce a suitable approximate subgradient with $\tilde{O}(1)$ evaluation queries, and show how to use such an approximate subgradient in the separation framework to produce a faster quantum algorithm.

Our new quantum algorithm for subgradient computation is the source of the quantum speedup of the entire algorithm and establishes a separation in query complexity for the subgradient computation between quantum ($\tilde{O}(1)$) and classical ($\tilde{\Omega}(n)$, see \sec{classical-gradient}) algorithms. This subroutine could also be of independent interest, in particular in the study of quantum algorithms based on gradient descent and its variants (e.g., \cite{rebentrost2016, kerenidis2017}).

Our techniques for finding an approximate subgradient only require an approximate oracle for the function to be differentiated. \thm{upper-main} also applies if the membership and evaluation oracles are given with error that is polynomially related to the required precision in minimizing the convex function (see \thm{mem-to-opt}). Precise definitions for these oracles with error can be found in \sec{oracle}.

On the other hand, we also aim to establish corresponding quantum lower bounds to understand the potential for quantum speedups for convex optimization. To this end, we prove:

\begin{theorem}\label{thm:lower-main}
There exists a convex body $K\subseteq \R^n$, a convex function $f$ on $K$, and a precision $\epsilon>0$, such that a quantum algorithm needs at least $\Omega(\sqrt{n})$ queries to a membership oracle for $K$ and $\Omega(\sqrt{n}/\log n)$ queries to an evaluation oracle for $f$ to output a point $\tilde{x}$ satisfying
\begin{align}\label{eqn:approx-all-contrib}
f(\tilde{x})\leq\min_{x\in K}f(x)+\epsilon
\end{align}
with high success probability (say, at least 0.8).
\end{theorem}

We establish the query lower bound on the membership oracle by reductions from search with wildcards \cite{ambainis2014quantum}. The lower bound on evaluation queries uses a similar reduction, but this only works for an evaluation oracle with low precision. To prove a lower bound on precise evaluation queries, we propose a discretization technique that relates the difficulty of the continuous problem to a corresponding discrete one. This approach might be of independent interest since optimization problems naturally have continuous inputs and outputs, whereas most previous work on quantum lower bounds focuses on discrete inputs. Using this technique, we can simulate one perfectly precise query by one low-precision query at discretized points, thereby establishing the evaluation lower bound as claimed in \thm{lower-main}. As a side point, this evaluation lower bound holds even for an unconstrained convex optimization problem on $\R^{n}$, which might be of independent interest since this setting has also been well-studied classically \cite{boyd2004optimization,nesterov2013introductory,nemirovsky1983problem,nemirovski1995information}.

We summarize our main results in \tab{main-OPT}.

\begin{table}[H]
\centering
\resizebox{0.75\columnwidth}{!}{%
\begin{tabular}{|c||c|c|}
\hline
 & Classical bounds & Quantum bounds {(this paper)} \\ \hline\hline
Membership queries & $\tilde{O}(n^{2})$ \cite{lee2017efficient}, $\Omega(n)$ \cite{lee2018personal} & $\tilde{O}(n)$, $\Omega(\sqrt{n})$ \\ \hline
Evaluation queries & $\tilde{O}(n^{2})$ \cite{lee2017efficient}, $\Omega(n)$ \cite{lee2018personal} & $\tilde{O}(n)$, $\tilde{\Omega}(\sqrt{n})$ \\ \hline\hline
Time complexity & $\tilde{O}(n^{3})$ \cite{lee2017efficient} & $\tilde{O}(n^{3})$ \\ \hline
\end{tabular}
}
\caption{Summary of classical and quantum complexities of convex optimization.}
\label{tab:main-OPT}
\end{table}

\subsection{Overview of techniques}

\subsubsection{Upper bound}\label{sec:upper-bound-techniques}

To prove our upper bound result in \thm{upper-main}, we use the well-known reduction from general convex optimization to the case of a linear objective function, which simplifies the problem to implementing an optimization oracle using queries to a membership oracle (\lem{sco}). For the reduction from optimization to membership, we follow the best known classical result in \cite{lee2017efficient} which implements an optimization oracle using $\softoh{n^2}$ membership queries and $\softoh{n^3}$ arithmetic operations. In \cite{lee2017efficient}, the authors first show a reduction from separation oracles to membership oracles that uses $\softoh{n}$ queries and then use a result from \cite{lee2015faster} to implement an optimization oracle using $\softoh{n}$ queries to a separation oracle, giving an overall query complexity of $\softoh{n^2}$.

The reduction from separation to membership involves the calculation of a \emph{height function} defined by the authors (see Eq.\ \eq{depth-function}), whose evaluation oracle can be implemented in terms of the membership oracle of the original set. A separating hyperplane is determined by computing a subgradient, which already takes $\softoh{n}$ queries. In fact, it is not hard to see that any classical algorithm requires $\tilde\Omega(n)$ classical queries (see \sec{classical-gradient}), so this part of the algorithm cannot be improved classically. The possibility of using the quantum Fourier transform to compute the gradient of a function using $\softoh{1}$ evaluation queries (\cite{jordan2005fast,gilyen2019optimizing}) suggests the possibility of replacing the subgradient procedure with a faster quantum algorithm. However, the techniques described in \cite{jordan2005fast,gilyen2019optimizing} require the function in question to have bounded second (or even higher) derivatives, and the height function is only guaranteed to be Lipschitz continuous (\defn{lipschitz}) and in general is not even differentiable.

To compute subgradients of general (non-differentiable) convex functions,  we introduce classical randomness (taking inspiration from \cite{lee2017efficient}) and construct a quantum subgradient algorithm that uses $\softoh{1}$ queries. Our proof of correctness (\sec{eval-to-grad}) has three main steps:
 \begin{enumerate}
 \item We analyze the average error incurred when computing the gradient using the quantum Fourier transform. Specifically, we show that this approach succeeds if the function has bounded second derivatives in the vicinity of the point where the gradient is to be calculated (see \algo{grad_est}, \algo{quantum-gradient}, and \lem{jordan-expected-bound}). Some of our calculations are inspired by \cite{gilyen2019optimizing}.
 \item We use the technique of \emph{mollifier functions} (a common tool in functional analysis \cite{hormander1990analysis}, suggested to us by \cite{lee2018personal} in the context of \cite{lee2017efficient}) to show that it is sufficient to treat infinitely differentiable functions (the mollified functions) with bounded first derivatives (but possibly large second derivatives). In particular, it is sufficient to output an approximate gradient of the mollified function at a point near the original point where the subgradient is to be calculated (see \lem{nearby-gradient}).
 \item We prove that convex functions with bounded first derivatives have second derivatives that lie below a certain threshold with high probability for a random point in the vicinity of the original point (\lem{effective-smoothness}). Furthermore, we show that a bound on the second derivatives can be chosen so that the smooth gradient calculation techniques work on a sufficiently large fraction of the neighborhood of the original point, ensuring that the final subgradient error is small (see \algo{classical-smoothing} and \thm{general-quantum-gradient}).
 \end{enumerate}

 The new quantum subgradient algorithm is then used to construct a separation oracle as in \cite{lee2017efficient} (and a similar calculation is carried out in \thm{gen-halfspace}). Finally the reduction from \cite{lee2015faster} is used to construct the optimization oracle using $\softoh{n}$ separation queries. From \lem{sco}, this shows that the general convex optimization problem can be solved using $\softoh{n}$ membership and evaluation queries and $\softoh{n^3}$ gates.

\subsubsection{Lower bound}\label{sec:lower-bound-intro}
We prove our quantum lower bounds on membership and evaluation queries separately before showing how to combine them into a single optimization problem. Both lower bounds work over $n$-dimensional hypercubes.

In particular, we prove both lower bounds by reductions from search with wildcards \cite{ambainis2014quantum}. In this problem, we are given an $n$-bit binary string $s$ and the task is to determine all bits of $s$ using wildcard queries that check the correctness of any subset of the bits of $s$: more formally, the input in the wildcard model is a pair $(T,y)$ where $T\subseteq\range{n}$ and $y\in\{0, 1\}^{|T|}$, and the query returns 1 if $s_{|T}=y$ (here the notation $s_{|T}$ represents the subset of the bits of $s$ restricted to $T$). Reference \cite{ambainis2014quantum} shows that the quantum query complexity of search with wildcards is $\Omega(\sqrt{n})$.

For our lower bound on membership queries, we consider a simple objective function, the sum of all coordinates $\sum_{i=1}^{n}x_{i}$. In other words, we take $c=\textbf{1}^{n}$ in \eq{rephrased-convex}. However, the position of the hypercube is unknown, and to solve the optimization problem (formally stated in \defn{sum-coordinate}), one must use the membership oracle to locate it.

Specifically, the hypercube takes the form $\bigtimes_{i=1}^{n}[s_{i}-2,s_{i}+1]$ (where $\bigtimes$ is the Cartesian product) for some offset binary string $s\in\{0,1\}^n$. In \sec{lower-membership}, we prove:
\begin{itemize}
\item Any query $x\in\R^{n}$ to the membership oracle of this problem can be simulated by one query to the search-with-wildcards oracle for $s$. To achieve this, we divide the $n$ coordinates of $x$ into four sets: $T_{x,0}$ for those in $[-2,-1)$, $T_{x,1}$ for those in $(1,2]$, $T_{x,\text{mid}}$ for those in $[-1,1]$, and $T_{x,\text{out}}$ for the rest. Notice that $T_{x,\text{mid}}$ corresponds to the coordinates that are always in the hypercube and $T_{x,\text{out}}$ corresponds to the coordinates that are always out of the hypercube; $T_{x,0}$ (resp., $T_{x,1}$) includes the coordinates for which $s_{i}=0$ (resp., $s_{i}=1$) impacts the membership in the hypercube. We prove in \sec{lower-membership} that a wildcard query with $T=T_{x,0}\cup T_{x,1}$ can simulate a membership query to $x$.

\item The solution of the sum-of-coordinates optimization problem explicitly gives $s$, i.e., it solves search with wildcards. This is because this solution must be close to the point $(s_{1}-2,\ldots,s_{n}-2)$, and applying integer rounding would recover $s$.
\end{itemize}
These two points establish the reduction of search with wildcards to the optimization problem, and hence establishes the $\Omega(\sqrt{n})$ membership quantum lower bound in \thm{lower-main} (see \thm{main-membership}).

For our lower bound on evaluation queries, we assume that membership is trivial by fixing the hypercube at $\mathcal{C}=[0,1]^{n}$. We then consider optimizing the max-norm function
\begin{align}\label{eqn:max-OPT-simple}
    f(x)=\max_{i\in\range{n}}|x_{i}-c_{i}|
\end{align}
for some unknown $c\in\{0,1\}^{n}$. Notice that learning $c$ is equivalent to solving the optimization problem; in particular, outputting an $\tilde{x}\in\mathcal{C}$ satisfying \eqn{approx} with $\eps=1/3$ would determine the string $c$. This follows because for all $i\in\range{n}$, we have $|\tilde{x}_{i}-c_{i}|\leq\max_{i\in\range{n}}|\tilde{x}_{i}-c_{i}|\leq 1/3$, and $c_{i}$ must be the integer rounding of $\tilde{x}_{i}$, i.e., $c_{i}=0$ if $\tilde{x}_{i}\in[0,1/2)$ and $c_{i}=1$ if $\tilde{x}_{i}\in[1/2,1]$. On the other hand, if we know $c$, then we know the optimum $x=c$.

We prove an $\Omega(\sqrt{n}/\log n)$ lower bound on evaluation queries for learning $c$. Our proof, which appears in \sec{lower-evaluation}, is composed of three steps:
\begin{enumerate}[label=\arabic*)]
\item We first prove a weaker lower bound with respect to the precision of the evaluation oracle. Specifically, if $f(x)$ is specified with $b$ bits of precision, then using binary search, a query to $f(x)$ can be simulated by $b$ queries to an oracle that inputs $(f(x),t)$ for some $t\in\R$ and returns 1 if $f(x)\leq t$ and returns 0 otherwise. We further without loss of generality assume $x\in[0,1]^{n}$. If $x\notin [0,1]^{n}$, we assign a penalty of the $L_{1}$ distance between $x$ and its projection $\pi(x)$ onto $[0,1]^{n}$; by doing so, $f(\pi(x))$ and $x$ fully characterizes $f(x)$ (see \eqn{projection-relationship}). Therefore, $f(x)\in[0,1]$, and $f(x)$ having $b$ bits of precision is equivalent to having precision $2^{-b}$.

Similar to the interval dividing strategy in the proof of the membership lower bound, we prove that one query to such an oracle can be simulated by one query to the search-with-wildcards oracle for $s$. Furthermore, the solution of the max-norm optimization problem explicitly gives $s$, i.e., it solves the search-with-wildcards problem. This establishes the reduction to search with wildcards, and hence establishes an $\Omega(\sqrt{n}/b)$ lower bound on the number of quantum queries to the evaluation oracle $f$ with precision $2^{-b}$ (see \lem{main-evaluation-weak}).

\item Next, we introduce a technique we call \emph{discretization}, which effectively simulates queries over an (uncountably) infinite set by queries over a discrete set. This technique might be of independent interest since proving lower bounds on functions with an infinite domain can be challenging.

We observe that the problem of optimizing \eqn{max-OPT-simple} has the following property: if we are given two strings $x,x'\in[0,1]^{n}$ such that $x_{1},\ldots,x_{n},1-x_{1},\ldots,1-x_{n}$ and $x'_{1},\ldots,x'_{n},1-x'_{1},\ldots,1-x'_{n}$ have the same ordering (for instance, strings $x=(0.1,0.2,0.7)$ and $x'=(0.1,0.3,0.6)$ both have the ordering $x_{1}\leq x_{2}\leq 1-x_{3}\leq x_{3}\leq 1-x_{2}\leq 1-x_{1}$), then
\begin{align}\label{eq:max-norm-argmax}
\arg\max_{i\in\range{n}}|x_{i}-c_{i}|=\arg\max_{i\in\range{n}}|x'_{i}-c_{i}|.
\end{align}
Furthermore, if $x'_{1},\ldots,x'_{n},1-x'_{1},\ldots,1-x'_{n}$ are $2n$ different numbers, then knowing the value of $f(x')$ implies the value of the $\arg\max$ in \eq{max-norm-argmax} (denoted $i^{*}$) and the corresponding $c_{i^{*}}$, and we can subsequently recover $f(x)$ given $x$ since $f(x)=|x_{i^{*}}-c_{i^{*}}|$. In other words, $f(x)$ can be computed given $x$ and $f(x')$.

Therefore, it suffices to consider all possible ways of ordering $2n$ numbers, rendering the problem discrete. Without loss of generality, we focus on $x'$ satisfying $\{x'_{1},\ldots,x'_{n},1-x'_{1},\ldots,1-x'_{n}\}=\{\frac{1}{2n+1},\ldots,\frac{2n}{2n+1}\}$, and we denote the set of all such $x'$ by $D_n$ (see also \eqn{discrete-set}). In \lem{discretization}, we prove that one classical (resp., quantum) evaluation query from $[0,1]^{n}$ can be simulated by one classical evaluation query (resp., two quantum evaluation queries) from $D_{n}$ using \algo{discretization}. To illustrate this, we give a concrete example with $n=3$ in \sec{discretization}.

\item Finally, we use discretization to show that one perfectly precise query to $f$ can be simulated by one query to $f$ with precision $\frac{1}{5n}$; in other words, $b$ in step 1) is at most $\lceil\log_{2} 5n\rceil=O(\log n)$ (see \lem{precision-simulation-formal}). This is because by discretization, the input domain can be limited to the discrete set $D_{n}$. Notice that for any $x\in D_{n}$, $f(x)$ is an integer multiple of $\frac{1}{2n+1}$; even if $f(x)$ can only be computed with precision $\frac{1}{5n}$, we can round it to the closest integer multiple of $\frac{1}{2n+1}$ which is exactly $f(x)$, since the distance $\frac{2n+1}{5n}<\frac{1}{2}$. As a result, we can precisely compute $f(x)$ for all $x\in D_{n}$, and thus by discretization we can precisely compute $f(x)$ for all $x\in [0,1]^{n}$.
\end{enumerate}
In all, the three steps above establish an $\Omega(\sqrt{n}/\log n)$ quantum lower bound on evaluation queries to solve the problem in Eq. \eqn{max-OPT-simple} (see \thm{main-membership}). In particular, this lower bound is proved for an unconstrained convex optimization problem on $\R^{n}$, which might be of independent interest.

As a side result, we prove that our quantum lower bound is optimal for the problem in \eqn{max-OPT-simple} (up to poly-logarithmic factors in $n$), as we can prove a matching $\tilde{O}(\sqrt{n})$ upper bound (\thm{main-evaluation-optimal}). Therefore, a better quantum lower bound on the number of evaluation queries for convex optimization would require studying an essentially different problem.

Having established lower bounds on both membership and evaluation queries, we combine them to give \thm{lower-main}. This is achieved by considering an optimization problem of dimension $2n$; the first $n$ coordinates compose the sum-of-coordinates function in \sec{lower-membership}, and the last $n$ coordinates compose the max-norm function in \sec{lower-evaluation}. We then concatenate both parts and prove \thm{lower-main} via reductions to the membership and evaluation lower bounds, respectively (see \sec{main-together}).

In addition, all lower bounds described above can be adapted to a convex body that is contained in the unit hypercube and that contains the discrete set $D_{n}$ to facilitate discretization; we present a ``smoothed'' hypercube (see \sec{smooth}) as a specific example.

\subsection{Open questions}

This work leaves several natural open questions for future investigation. In particular:
\begin{itemize}
\item Can we close the gap for both membership and evaluation queries? Our upper bounds on both oracles in \thm{upper-main} uses $\tilde{O}(n)$ queries, whereas the lower bounds of \thm{lower-main} are only $\tilde{\Omega}(\sqrt{n})$.

\item Can we improve the time complexity of our quantum algorithm? The time complexity $\tilde{O}(n^{3})$ of our current quantum algorithm matches that of the classical state-of-the-art algorithm \cite{lee2017efficient} since our second step, the reduction from optimization to separation, is entirely classical. Is it possible to improve this reduction quantumly?

\item What is the quantum complexity of convex optimization with a first-order oracle (i.e., with direct access to the gradient of the objective function)? This model has been widely considered in the classical literature (see for example Ref. \cite{nesterov2013introductory}).
\end{itemize}

\paragraph{Organization.}
Our quantum upper bounds are given in \sec{upper-bound} and lower bounds are given in \sec{lower-bound}. Appendices present auxiliary lemmas (\sec{auxiliary}) and proof details for upper bounds (\sec{upper-bound-appendix}) and lower bounds (\sec{lower-bound-appendix}), respectively.

\paragraph{Related independent work.}
In independent simultaneous work, van Apeldoorn, Gily{\'e}n, Gribling, and de Wolf \cite{vanApeldoorn2018optimization} establish a similar upper bound, showing that $\tilde{O}(n)$ quantum queries to a membership oracle suffice to optimize a linear function over a convex body (i.e., to implement an optimization oracle). Their proof follows a similar strategy to ours, using a quantum algorithm for evaluating gradients in $\tilde{O}(1)$ queries to implement a separation oracle. As in our approach, they use a randomly sampled point in the neighborhood of the point where the subgradient is to be calculated. The only major difference is that they use finite approximations of the gradient and second derivatives, whereas we use these quantities in their original form and give an argument based on mollifier functions to ensure that they are well defined.

Reference \cite{vanApeldoorn2018optimization} also establishes quantum lower bounds on the query complexity of convex optimization, showing in particular that $\Omega(\sqrt{n})$ quantum queries to a separation oracle are needed to implement an optimization oracle, implying an $\Omega(\sqrt{n})$ quantum lower bound on the number of membership queries required to optimize a convex function. While Ref.\ \cite{vanApeldoorn2018optimization} does not explicitly focus on evaluation queries, those authors have pointed out to us that an $\Omega(\sqrt{n})$ lower bound on evaluation queries can be obtained from their lower bound on membership queries (although our approach gives a bound with a better Lipschitz parameter).

\section{Upper bound}
\label{sec:upper-bound}

In this section, we prove:

  \begin{theorem}
    \label{thm:upper-bound-informal}
    An optimization oracle for a convex set $K \subseteq \R^n$ can be implemented using $\tilde{O}(n)$ quantum queries to a membership oracle for $K$, with gate complexity $\tilde{O}(n^3)$.
  \end{theorem}

  The following lemma shows the equivalence of optimization oracles to a general convex optimization problem.

    \begin{lemma}
  \label{lem:sco}
  Suppose a reduction from an optimization oracle to a membership oracle for convex sets requires $O(g(n))$ queries to the membership oracle. Then the problem of optimizing a convex function over a convex set can be solved using $O(g(n))$ queries to both the membership oracle and the evaluation oracle.
\end{lemma}
\begin{proof}
  The problem $\min_{x \in K} f(x)$ reduces to the problem $\min_{(x',x) \in K'} x'$ where $K'$ is defined as in (\ref{prb:gen_opt}). $K'$ is the intersection of convex sets and is therefore itself convex.  A membership oracle for $K'$ can be implemented using 1 query each to the membership oracle for $K$ and the evaluation oracle for $f$. Since $O(g(n))$ queries to the membership oracle for $K'$ are sufficient to optimize any linear function, the result follows.
\end{proof}

\thm{upper-main} directly follows from \thm{upper-bound-informal} and \lem{sco}.

\paragraph{Overview.}
This part of the paper is organized following the plan outlined in \sec{upper-bound-techniques}. Precise definitions of oracles and other relevant terminology appear in \sec{oracle}. \sec{eval-to-grad} develops a fast quantum subgradient procedure that can be used in the classical reduction from optimization to membership. This is done in two parts:
\begin{enumerate}
\item \sec{gradient-smooth} presents an algorithm based on the quantum Fourier transform that calculates the gradient of a function with bounded second derivatives (i.e., a $\beta$-smooth function) with bounded expected one-norm error.
\item \sec{gradient-general} uses mollification to restrict the analysis to infinitely differentiable functions without loss of generality, and then uses classical randomness to eliminate the need for bounded second derivatives.
\end{enumerate}
In \sec{mem-to-sep} we show that the new quantum subgradient algorithm fits into the classical reduction from \cite{lee2017efficient}. Finally, we describe the reduction from optimization to membership in \sec{sep-to-opt}.

\subsection{Oracle definitions}\label{sec:oracle}

In this section, we provide precise definitions for the oracles for convex sets and functions that we use in our algorithm and its analysis. We also provide precise definitions of Lipschitz continuity and $\beta$-smoothness, which we will require in the rest of the section.

\begin{definition}[Ball in $L_p$ norm]
\label{defn:ball}
The ball of radius $r>0$ in $L_p$ norm $\norm{\cdot}_p$ centered at $x \in \R^n$ is $B_p(x,r) := \{y \in \R_n \mid \norm{x-y}_p \le r\}$.
\end{definition}

\begin{definition}[Interior of a convex set]
  \label{defn:conv_interior}
  For any $\delta > 0$, the $\delta$-interior of a convex set $K$ is defined as
  $B_2(K,-\delta) := \lbrace x \mid B_2(x,\delta) \subseteq K \rbrace$.
\end{definition}

\begin{definition}[Neighborhood of a convex set]
  \label{defn:conv_nbd}
  For any $\delta > 0$, the $\delta$-neighborhood of a convex set $K$ is defined as
  $B_{2}(K,\delta) := \lbrace x\mid\exists\,y \in K\text{ s.t. } \norm{x-y}_2 \le \delta \rbrace$.
\end{definition}

\begin{definition}[Evaluation oracle]
  \label{defn:eval_oracle}
  When queried with $x \in \R^n$ and $\delta > 0$, output $\alpha$ such that $|\alpha - f(x)| \le \delta$. We use $\EVAL_\delta(f)$ to denote the time complexity. The classical procedure or quantum unitary representing the oracle is denoted by $O_f$.
\end{definition}

\begin{definition}[Membership oracle]
  \label{defn:mem_oracle}
  When queried with $x \in \R^n$ and $\delta > 0$, output an assertion that $x \in B_2(K,\delta)$ or $x \notin B_2(K,-\delta)$. The time complexity is denoted by $\MEM_\delta(K)$. The classical procedure or quantum unitary representing the membership oracle is denoted by $O_K$.
\end{definition}

\begin{definition}[Separation oracle]
  \label{defn:sep_oracle}
  When queried with $x \in \R^n$ and $\delta > 0$, with probability $1 - \delta$, either
  \begin{itemize}[nosep]
  \item assert $x \in B_2(K,\delta)$ or
  \item output a unit vector $\hat{c}$ such that $\hat{c}^Tx \le \hat{c}^Ty + \delta$ for all $y\in B_2(K,-\delta)$.
  \end{itemize}
  The time complexity is denoted by $\SEP_\delta(K)$.
\end{definition}

\begin{definition}[Optimization oracle]
  \label{defn:opt_oracle}
  When queried with a unit vector $c$, find $y \in \R^n$ such that $c^Tx \le c^Ty + \delta$ for all $x \in B_2(K,-\delta)$ or asserts that $B_2(K,\delta)$ is empty. The time complexity of the oracle is denoted by $\OPT_\delta(K)$.
\end{definition}

\begin{definition}[Subgradient]
  \label{defn:subgradient}
  A subgradient of a convex function $f \colon \R^n \to \R$ at $x$, is a vector $g$ such that
  \begin{equation}
    \label{eq:exact-subgradient}
    f(y) \ge f(x) + \langle g, y-x \rangle
  \end{equation}
  for all $y \in \R^n$. For a differentiable convex function, the gradient is the only subgradient. The set of subgradients of $f$ at $x$ is called the subdifferential at $x$ and denoted by $\partial f(x)$.
\end{definition}

\begin{definition}[$L$-Lipschitz continuity]
  \label{defn:lipschitz}
  A function $f \colon \R^n \to \R$ is said to be $L$-Lipschitz continuous (or simply $L$-Lipschitz) in a set $S$ if for all $x \in S$, $\norm{g}_\infty \le L$ for any $g \in \partial f(x)$. An immediate consequence of this is that for any $x,y \in S$,
  \begin{equation}
    \label{eq:lipschitz-cond}
      |f(y) - f(x)| \leq L\|y - x\|_{\infty}.
  \end{equation}
\end{definition}

\begin{definition}[$\beta$-smoothness]
  \label{defn:beta-smooth}
  A function $f \colon \R^n \to \R$ is said to be $\beta$-smooth in a set $S$ if for all $x \in S$, the magnitudes of the second derivatives of $f$ in all directions are bounded by $\beta$. This also means that the largest magnitude of an eigenvalue of the Hessian $\nabla^2f(x)$ is at most $\beta$. Consequently, for any $x,y \in S$, we have
  \begin{equation}
    \label{eq:beta-smooth-cond}
    f(y) \le f(x) + \langle \nabla f(x), y - x \rangle + \frac{\beta}{2}\|y - x\|_{\infty}^{2}.
  \end{equation}
\end{definition}

\subsection{Evaluation to subgradient}
\label{sec:eval-to-grad}

In this section we present a procedure that, given an evaluation oracle for an $L$-Lipschitz continuous function $f\colon \R^n \to \R$ with evaluation error at most $\epsilon > 0$, a point $x \in \R^n$, and an ``approximation scale'' factor $r_1 > 0$, computes an approximate subgradient $\tilde{g}$ of $f$ at $x$. Specifically, $\tilde{g}$ satisfies
\begin{align}
  \label{eq:approx-subgrad}
  f(q) \ge f(x) + \langle \tilde{g}, q - x \rangle - \zeta\norm{q-x}_\infty - 4nr_1L
\end{align}
for all $q \in \R^n$, where $\E\zeta \le \xi(r_1,\epsilon)$ and $\xi$ must monotonically increase with $\epsilon$ as $\epsilon^\alpha$ for some $\alpha > 0$.
Here $\zeta$ is the error in the subgradient that is bounded in expectation by the function $\xi$.

\subsubsection{Smooth functions}
\label{sec:gradient-smooth}

We first describe how to approximate the gradient of a smooth function. \algo{grad_est} and \algo{quantum-gradient} use techniques from \cite{jordan2005fast} and \cite{gilyen2019optimizing} to evaluate the gradient of a function with bounded second derivatives in the neighborhood of the evaluation point. To analyze their behavior, we begin with the following lemma showing that \algo{grad_est} provides a good estimate of the gradient with bounded failure probability.

\begin{algorithm}[ht]
  \caption{$\algname{GradientEstimate}(f,\epsilon,L,\beta,x_0)$}
  \label{algo:grad_est}
  \KwData{Function $f$, evaluation error $\epsilon$, Lipschitz constant $L$, smoothness parameter $\beta$, and point $x_0$.\\
    Define
    \begin{itemize}[nosep]
      \item $l = 2\sqrt{{\epsilon}/{n\beta}}$ to be the size of the grid used,
      \item $b \in \N$ such that $\frac{24\pi\sqrt{n\epsilon \beta}}{L} \le \frac{1}{2^b} = \frac{1}{N} \le \frac{48\pi\sqrt{n\epsilon \beta}}{L}$,
      \item $b_0 \in \N$ such that $\frac{N\epsilon}{2Ll} \le \frac{1}{2^{b_0}} = \frac{1}{N_0} \le \frac{N\epsilon}{Ll}$,
      \item $F(x) = \frac{N}{2Ll}[f(x_0 + \frac{l}{N}(x - N/2)) - f(x_0)]$, and,
      \item $ \gamma: \lbrace 0,1,\dots,N-1 \rbrace \to G:=\lbrace -N/2,-N/2+1,\dots,N/2-1 \rbrace$ s.t. $\gamma(x) = x - N/2$.
    \end{itemize}
    Let $O_F$ denote a unitary operation acting as $O_F\ket{x} = e^{2\pi i \tilde{F}(x)}\ket{x}$, where $|\tilde{F}(x) - F(x) | \le \frac{1}{N_0}$, with $x$ represented using $b$ bits and $\tilde{F}(x)$ represented using $b_0$ bits.
  }
  Start with $n$ $b$-bit registers set to 0 and Hadamard transform each to obtain
  \begin{align}
    \label{eq:grid}
    \frac{1}{\sqrt{N^n}}\sum_{x_1,\ldots,x_n \in \{0,1,\ldots,N-1\}}\ket{x_1,\ldots,x_n};
  \end{align}

  Perform the operation $O_F$ and the map $\ket{x} \mapsto \ket{\gamma(x)}$ to obtain
  \begin{align}
    \frac{1}{N^{n/2}}\sum_{g \in G^n}e^{2\pi i \tilde{F}(g)} \ket{g};
  \end{align}

  Apply the inverse QFT over $G$ to each of the registers\;
  Measure the final state to get $k_1,k_2,\ldots,k_n$ and report $\tilde{g} = \frac{2L}{N}(k_1,k_2,\ldots,k_n)$ as the result. \label{lin:grad_est_measure}
\end{algorithm}

\begin{lemma}
  \label{lem:jordan_whp}
  Let $f\colon \R^n \to \R$ be an $L$-Lipschitz function that is specified by an evaluation oracle with error at most $\epsilon$. Let $f$ be $\beta$-smooth in $B_\infty(x,2\sqrt{{\epsilon}/{\beta}})$, and let $\tilde{g}$ be the output of $\algname{GradientEstimate}(f,\epsilon,L,\beta,x_0)$ (from \algo{grad_est}). Then
  \begin{align}
    \Pr\left[ |\tilde{g}_i - \nabla f(x)_i| > 1500\sqrt{n\epsilon\beta} \right] < \frac{1}{3}, \quad \forall\,i \in \range{n}.
  \end{align}
\end{lemma}

The proof of \lem{jordan_whp} is deferred to \lem{jordan_whp_detailed} in the appendix.

Next we analyze \algo{quantum-gradient}, which uses several calls to \algo{grad_est} to provide an estimate of the gradient that is close in expected $L_1$ distance to the true value.

\begin{algorithm}[htbp]
  \caption{$\algname{SmoothQuantumGradient}(f,\epsilon,L,\beta,x)$}
  \label{algo:quantum-gradient}
 \KwData{Function $f$, evaluation error $\epsilon$, Lipschitz constant $L$, smoothness parameter $\beta$, and point $x$.}
  Set $T$ such that $2e^{{-T^2}/{24}} \le {750\sqrt{n\epsilon\beta}}/{L}$\;
  \For{$t=1,2,\ldots,T$}
  {
    $e^{(t)} \leftarrow \algname{GradientEstimate}(f,\epsilon,L,\beta,x)$\;
  }
  \For{$i=1,2,\ldots,n$}
  {
    If more than $T/2$ of $e_i^{(t)}$ lie in an interval of size $3000\sqrt{n\epsilon\beta}$, set $\tilde{g}_i$ to be the median of the points in that interval\;
    Otherwise, set $\tilde{g}_i = 0$\;
  }
  Output $\tilde{g}$.
\end{algorithm}

\begin{lemma}
  \label{lem:jordan-expected-bound}
  Let $f$ be a convex, $L$-Lipshcitz continuous function that is specified by an evaluation oracle with error at most $\epsilon$. Suppose $f$ is $\beta$-smooth in $B_\infty(x,2\sqrt{{\epsilon}/{\beta}})$. Let
\begin{align}
\tilde{g} = \algname{SmoothQuantumGradient}(f,\epsilon,L,\beta,x)
\end{align}
(from \algo{quantum-gradient}). Then for any $i \in \range{n}$, we have $|\tilde{g}_i| \le L$ and $\E|\tilde{g}_i - \nabla f(x)_i| \le 3000\sqrt{n\epsilon\beta}$; hence
  \begin{align}
    \E \norm{\tilde{g} - \nabla f(x)}_1 \le 3000n^{3/2}\sqrt{\epsilon\beta}.
    \label{eq:jordan-expected-bound}
  \end{align}
  If $L$, $1/\beta$, and $1/\epsilon$ are $\poly(n)$, the $\algname{SmoothQuantumGradient}$ algorithm uses $\softoh{1}$ queries to the evaluation oracle and $\softoh{n}$ gates.
\end{lemma}

\begin{proof}
  For each dimension $i\in\range{n}$ and each iteration $t \in \range{T}$, consider the random variable
  \begin{align}
    X_i^t = \begin{cases}
     1 & \text{if $|e_i^{(t)} - \nabla f(x)_i| > 1500\sqrt{n\epsilon\beta}$} \\
     0 & \text{otherwise}.
   \end{cases}
  \end{align}
From the conditions on the function $f$, \lem{jordan_whp} applies to $\algname{GradientEstimate}(f,\epsilon,L,\beta,x)$, and thus $\Pr(X_i^t = 1) < 1/3$. Thus, by the Chernoff bound, $\Pr\left[|\tilde{g}_i- \nabla f(x)_i| \le 1500\sqrt{n\epsilon\beta}\right] >1 - 2e^{-{T^2}/{24}} \ge 1 - {750\sqrt{n\epsilon\beta}}/{L}$.
  In the remaining cases, $|\tilde{g}_i - \nabla f(x)_i| \le 2L$ (see \lin{grad_est_measure} of \algo{grad_est}). Thus $\E| \tilde{g}_i - \nabla f(x)_i| \le 3000\sqrt{n\epsilon\beta}$ for all $i\in\range{n}$, and \eq{jordan-expected-bound} follows.

  The algorithm makes $T=\poly(\log(1/n\epsilon\beta))$ calls to a procedure that makes one query to the evaluation oracle. Thus the query complexity is $\softoh{1}$.
  To evaluate the gate complexity, observe that we iterate over $n$ dimensions, using $\poly(b) = \poly(\log(1/n\epsilon\beta))$ gates for the quantum Fourier transform over each. This process is repeated $T=\poly(\log(1/n\epsilon\beta))$ times. Thus the entire algorithm uses $\softoh{n}$ gates.
\end{proof}

\subsubsection{Extension to non-smooth functions}
\label{sec:gradient-general}

Now consider a general $L$-Lipschitz continuous convex function $f$. We show that any such function is close to a smooth function, and we consider the relationship between the subgradients of the original function and the gradient of its smooth approximation.

For any $\delta>0$, let $m_\delta \colon \R^n \to \R$ be the \emph{mollifier function of width $\delta$}, defined as
\begin{align}
  \label{eq:mollifier}
  m_\delta(x) := \begin{cases}
    \frac{1}{I_n}\exp\left(-\frac{1}{1 - \norm{x/\delta}_2^2}\right)
    & x \in B_2(0,\delta) \\
    0 & \text{otherwise,}
  \end{cases}
\end{align}
where $I_n$ is chosen such that $\int_{B_2(0,\delta)} m_{\delta}(x) \,\d^n{x} = 1$.
The mollification of $f$, denoted $F_\delta := f \ast m_\delta$, is obtained by convolving it with the mollifier function, i.e.,
\begin{equation}
  \label{eq:mollified}
  F_\delta(x)
  = (f \ast m_\delta)(x)
  = \int_{\R^n} f(x-y) m_\delta(y) \, \d^n{x}.
\end{equation}
The mollification of $f$ has several key properties, as follows:

\begin{proposition}\label{prop:mollify}
  Let $f\colon \R^n \to \R$ be an $L$-Lipschitz convex function with mollification $F_\delta$. Then
  \begin{enumerate}[nosep,label=\normalfont(\roman*)]
  \item $F_\delta$ is infinitely differentiable, \label{itm:mollify_smooth1}
  \item $F_\delta$ is convex, \label{itm:mollify_convex}
  \item $F_\delta$ is $L$-Lipschitz continuous, and \label{itm:mollify_lip}
  \item $|F_\delta(x) - f(x)| \le L\delta$. \label{itm:mollify_close}
  \end{enumerate}
\end{proposition}
These properties of the mollifier function are well known in functional analysis \cite{hormander1990analysis}. For completeness a proof is provided in \lem{mollify}.

Furthermore, an approximate gradient of the mollified function gives an approximate subgradient of the original function, as quantified by the following lemma.

\begin{lemma}
  \label{lem:nearby-gradient}
  Let $f\colon \R^n \to \R$ be an infinitely differentiable $L$-Lipschitz continuous convex function with mollification $F_\delta$. Then any $\tilde{g}$ satisfying $\norm{\tilde{g} - \nabla F_\delta(y)}_1 = \zeta$ for some $y \in B_\infty(x,r_1)$
  satisfies
  \begin{equation}
    \label{eq:nearby-gradient}
    f(q) \ge f(x) + \langle \tilde{g}, q - x \rangle - \zeta\norm{q-x}_\infty - 4nr_1L - 2L\delta.
  \end{equation}
  Here $\zeta$ is the error in the subgradient and $\delta$ is the parameter used in the mollifier function.
\end{lemma}
\begin{proof}
  For all $q \in \R^n$, convexity of $F_\delta$ implies
  \begin{align}
  \label{eq:grad-subgrad}
  F_\delta(q)
  &\ge F_\delta(y) + \langle \nabla F_\delta(y), q-y \rangle \\
  &=   F_\delta(x) + \langle \nabla F_\delta(y), q - x \rangle + \langle \nabla F_\delta(y), x - y \rangle + (F_\delta(y) - F_\delta(x)) \\
  &\ge F_\delta(x) + \langle \nabla F_\delta(y), q-x \rangle - 4nr_1L \\
  &\ge F_\delta(x) + \langle \tilde{g}, q-x \rangle - \zeta\norm{q-x}_\infty - 4nr_1L,
  \end{align}
  so \eq{nearby-gradient} follows from \prop{mollify}\itm{mollify_close}.
\end{proof}

Now consider $\delta$ such that $L\delta \ll \epsilon$. Then the evaluation oracle with error $\epsilon$ for $f$ is also an evaluation oracle for $F_\delta$ with error $\epsilon + L\delta \approx \epsilon$. Thus the given evaluation oracle is also the evaluation oracle for an infinitely differentiable convex function with the same Lipschitz constant and almost the same error, allowing us to analyze infinitely differentiable functions without loss of generality (as long as we make no claim about the second derivatives). This idea is made precise in \thm{general-quantum-gradient}. (Note that the mollification of $f$ is never computed or estimated by our algorithm; it is only a tool for analysis.)

Unfortunately, \lem{jordan-expected-bound} cannot be directly used to calculate subgradients for $F_\delta$ as $\delta \to 0$. This is because there exist convex functions (such as $f(x) = |x|$) where if $|f(x) - g(x)| \le \delta$ and $g(x)$ is $\beta$-smooth, then $\beta\delta \ge c $ for some constant $c$ (see \lem{non-smooth-differentiable} in the appendix). Thus using the \algname{SmoothQuantumGradient} algorithm at $x=0$ will give us a one-norm error of $ 3000 n^{3/2}\sqrt{\epsilon\beta} \ge 3000 n^{3/2} \sqrt{c}$, which is independent of $\epsilon$.

To fix this issue, we take inspiration from \cite{lee2017efficient} and introduce classical randomness into the gradient evaluation. In particular, the following lemma shows that for a Lipschitz continuous function, if we sample at random from the neighborhood of any given point, the probability of having large second derivatives is small. Let $y \sim Y$ indicate that $y$ is sampled uniformly at random from the set $Y$. Also, let $\lambda(x)$ be the largest eigenvalue of the Hessian matrix $\nabla^2f(x)$ at $x$. Since the Hessian is positive semidefinite, we have $\lambda(x) \le \Delta f(x) := \Tr(\nabla^2 f(x))$. Thus the second derivatives of $f$ are upper bounded by $\Delta f(x)$.

  Let $\eta(y)$ denote the area element on the surface $\partial B_\infty(x,r_1)$, defined as
  \begin{align}
    \label{eq:area-element}
    \eta(y)_i :=
    \begin{cases}
      1 & \text{if $y_i - x_i \ge r_1$} \\
      0 & \text{otherwise}.
    \end{cases}
  \end{align}
We have
  \begin{align}
    \label{eq:avg-laplacian}
    \E_{y \sim B_\infty(x,r_1)} \Delta f(y)&= \frac{1}{(2r_1)^n}\int\limits_{B_\infty(x,r_1)} \Delta f(y) \,\d^n{y} \\
                                          &= \frac{1}{(2r_1)^n}\int\limits_{\partial B_\infty(x,r_1)} \langle \nabla f(y), \eta(y) \rangle \, \d^{n-1}{y} \label{eq:avg-laplacian-2} \\
    & \le \frac{1}{(2r_1)^n}(2n)(2r_1)^{n-1}L = \frac{nL}{r_1}
  \end{align}
  where \eq{avg-laplacian-2} comes from the divergence theorem (the integral of the divergence of a vector field over a set is equal to the integral of the vector field over the surface of the set). This indicates that while the second derivatives of a Lipschitz continuous function can be unbounded at individual points, its expected value for a point uniformly sampled in an extended region is bounded.

Now, consider a grid of side length $l$ (aligned with the coordinate axes) embedded in $B_\infty(x,r_1)$. We denote this grid by $\grid{B_\infty(x,r_1)}{l}$. For any $i \in [n]$ and a $y$ sampled uniformly from $\grid{B_\infty(x,r_1)}{l}$, the expectation of the integral of the second directional derivative in the $i^{\mathrm{th}}$ coordinate direction over a segment from $y$ to the point $y + l e_i$ is
\begin{align}
  \label{eq:bound-on-average-single}
  \E_{y \sim \grid{B_\infty(x,r_1)}{l}}\left[\int_{y_i}^{y_i+ le_i} \frac{\d^2f(z)}{\d z_i^2}\, \d z_i\right] = \E_{y \sim \grid{B_\infty(x,r_1)}{l}}\left[\int_{0}^{l}\frac{\d^2f(y + te_i)}{dt^2}\d t\right]\le \frac{Ll}{r_1}.
\end{align}

To see this, note that there are $2r_1/l$ segments of length $l$ (corresponding to different points $y$) inside $B_\infty(y,r_1)$. The total integral of the directional derivative over these segments is upper bounded by the change in the $i^{\mathrm{th}}$ component of the gradient, which is in turn bounded by $2L$ due to the Lipschitz property of $f$.

Let $\devlin\colon \R^n \times \R^n \to \R$ defined by
\begin{align}
  \devlin(y,z) := \big\rvert f(z) - f(y) - \langle \nabla f(y), z - y \rangle \big\rvert
\end{align}
be a function that quantifies the deviation from linearity of $f$ between $y$ and $z$.
We now show the following lemma that bounds this deviation in the neighborhood of a randomly sampled point (with high probability).
\begin{lemma}
  \label{lem:effective-smoothness}
  Let $f \colon \R^n \to \R$ be an $L$-Lipschitz continuous, infinitely differentiable, convex function. Then for a point $y$ chosen uniformly from $\grid{B_\infty(x,r_1)}{l}$, and any $p \in \R$ such that $p \geq n$,
  \begin{align}
    \label{eq:effective-smoothness}
    \devlin(y,z) \le \frac{pnl^2L}{r_1}, \quad \forall\,z\in B_\infty(y,l)
  \end{align}
  with probability at least $1 - \frac{n}{p}$.
\end{lemma}

\begin{proof}
  Note that $\devlin(y,z)$ is a convex function of $z$ and must attain its maximum at one of the extremal points (vertices) of the hypercube $B_\infty(y,l)$, which are the $2^n$ points of the form
  \begin{align}\label{eq:hypercube-vertices}
    \big\{y + ls \mid s \in \{-1,1\}^n\big\}.
  \end{align}
This is because every point in the hypercube is a convex combination of the vertices, so having a higher function value at an internal point than at all the vertices would violate convexity.

  Consider a path from $y$ to a vertex of $B_\infty(y,l)$ consisting of $n$ segments of length $l$ aligned along the $n$ coordinate axes. For example, the path could move a distance $l$ along $\pm e_1,\pm e_2,\dots,\pm e_n$ until the vertex is reached. Using Markov's inequality with \eq{bound-on-average-single}, we have for every coordinate direction $i \in [n]$,
\begin{align}
  \label{eq:bound-on-average}
  \Pr_{y \sim \grid{B_\infty(x,r_1)}{l}}\left[\int_y^{y+ le_i} \frac{\d^2f(z)}{\d z_i^2}\, \d z > \frac{pLl}{r_1} \right] \le \frac{1}{p}.
\end{align}
Thus with probability at least $1 - \frac{1}{p}$, the increase in the deviation from linearity along each segment, as quantified by the function $\Delta$, is at most $\frac{pl^2L}{r_1}$. Using the union bound, with probability at least $1 - \frac{n}{p}$, the total deviation from linearity along the path is at most $\frac{pnl^2L}{r_1}$ as claimed.
\end{proof}

\lem{effective-smoothness} shows that with high probability a sampled point in $B_\infty(x,r_1)$ has a deviation from linearity in its neighborhood which is the same as that for a function with smoothness parameter $\frac{2pL}{r_1}$. The analysis of the gradient estimation procedure (\lem{jordan-expected-bound}) uses the smoothness of the function only to bound its deviation from linearity. Thus, \algo{quantum-gradient} can be applied as if to a function with smoothness parameter $\frac{2pL}{r_1}$. This observation is applied to find an approximate subgradient in \algo{classical-smoothing}.

\begin{algorithm}[htbp]
  \caption{$\algname{QuantumSubgradient}(f,\epsilon,L,x,r_1)$}
  \label{algo:classical-smoothing}
 \KwData{Function $f$, evaluation error $\epsilon$, Lipschitz constant $L$, point $x \in \R^n$, length $r_1 > 0$.}
 Sample $y \sim \grid{B_\infty(x,r_1)}{l}$\;
 Output $\tilde{g} = \algname{SmoothQuantumGradient}(f,\epsilon,L,{2n^{1/3}L}/{r_1^{2/3}\epsilon^{1/3}},y)$.
\end{algorithm}

\begin{theorem}
  \label{thm:general-quantum-gradient}
  Let $f$ be a convex, $L$-Lipschitz function that is specified by an evaluation oracle with error $\epsilon < \min\{1, r_1/n^2\}$. Let $\tilde{g} = \algname{QuantumSubgradient}(f,\epsilon,L,x,r_1)$ (from \algo{classical-smoothing}). Then for all $q \in \R^n$,
  \begin{align}
    f(q) \ge f(x) + \langle \tilde{g}, q - x \rangle - \zeta\norm{q-x}_\infty - 4nr_{1}L,
  \end{align}
  where $\E\zeta \le \frac{5000Ln^{5/3}\epsilon^{1/3}}{r_1^{1/3}}$.
\end{theorem}

\begin{proof}
Consider $F_\delta$ such that $L\delta \ll \epsilon$. From \prop{mollify}, $F_\delta$ is infinitely differentiable, convex, and $L$-Lipschitz. The given evaluation oracle for $f$ is also an evaluation oracle for $F_\delta$ with error $\epsilon_1 = \epsilon + L\delta$.

Assume without loss of generality that $L \ge 1$ (if not, the algorithm can be run with $L = 1$). Set $p = r_1^{1/3}n^{1/3}/\epsilon_1^{1/3}$ ($n/p < 1$ by the assumption on $\epsilon$). As observed above, \lem{effective-smoothness} shows that for $y \in \grid{B_\infty(x,r_1)}{l}$, \algo{quantum-gradient} produces $g = \algname{SmoothQuantumGradient}(F_\delta,\epsilon_1,L,\frac{2pL}{r_1},y)$ correctly with probability at least $1 - {n}/{p}$.
  Thus for each $i \in \range{n}$, we have:
  \begin{enumerate}
  \item With probability at least $1 - {n^{2/3}\epsilon_1^{1/3}}/{r_1^{1/3}}$,
    \begin{align}
      \E|g_i - \nabla F_\delta(y)_i| \le 3000\sqrt{\frac{2n\epsilon_1 pL}{r_1}} \le 3000L\sqrt{\frac{2n\epsilon_1 p}{r_1}} = \frac{3000\sqrt{2}Ln^{2/3}\epsilon_1^{1/3}}{r_1^{1/3}};
    \end{align}
  \item With probability at most $n/p = {n^{2/3}\epsilon_1^{1/3}}/{r_1^{1/3}}$, the algorithm fails. From Lipschitz continuity, $|\nabla F_\delta(x)_i| \le L$, and from \lem{jordan-expected-bound}, $|g_i| \le L$. Therefore,
    \begin{align}
      \E|g_i - \nabla F_\delta(y)_i| \le 2L.
    \end{align}
  \end{enumerate}
  Finally, we have
  \begin{equation}
    \label{eq:expected-component}
    \E_{y \sim \grid{B_\infty(x,r_1)}{l}} |g_i - \nabla F_\delta(y)_i| \le \frac{3000\sqrt{2}Ln^{2/3}\epsilon_1^{1/3}}{r_1^{1/3}} + \frac{2Ln}{p} < \frac{5000Ln^{2/3}\epsilon_1^{1/3}}{r_1^{1/3}},
  \end{equation}
hence
  \begin{equation}
    \label{eq:expected-onenorm}
    \E_{y \sim \grid{B_\infty(x,r_1)}{l}} \norm{g - \nabla F_\delta(y)}_1 \le \frac{5000Ln^{5/3}\epsilon_1^{1/3}}{r_1^{1/3}}.
  \end{equation}
  Thus from \lem{nearby-gradient},
  \begin{equation}
    f(q) \ge f(x) + \langle g, q - x \rangle - \zeta\norm{q-x}_\infty - 4nr_1L - 2L\delta
  \end{equation}
  for all $q \in \R^n$ where $\E\zeta \le \frac{5000Ln^{5/3}\epsilon_1^{1/3}}{r_1^{1/3}}$.
  Now let $\delta \to 0$. Then $F_\delta \to f$, $\epsilon_1 \to \epsilon$, and $g \to \tilde{g}$. Finally,
  \begin{equation}
    f(q) \ge f(x) + \langle \tilde{g}, q - x \rangle - \zeta\norm{q-x}_\infty - 4nr_1L
  \end{equation}
  for all $q \in \R^n$, where $\E\zeta \le \frac{5000Ln^{5/3}\epsilon^{1/3}}{r_1^{1/3}}$.
\end{proof}

\subsection{Membership to separation}
\label{sec:mem-to-sep}

\begin{algorithm}
  \caption{$\algname{SeparatingHalfspace}(K,p,\rho,\delta)$}
  \label{algo:halfspace}
  \KwData{Convex set $K$ such that $B_2(0,r) \subset K \subset B_2(0,R), \kappa=R/r$, $\delta$-precision membership oracle for $K$, point $p$.}
  \If {\text{the membership oracle asserts that} $p \in B_2(K,\delta)$} {
    \textbf{Output:} $p \in B_2(K,\delta)$. \label{lin:halfspace_contained}
  }
  \uElseIf{$ p \notin B_2(0,R)$} {
    \textbf{Output:} the halfspace $\{x \in \R^n \mid 0 > \langle x-p,p \rangle\}$.
  }
  \Else {
    Define $h_p(x)$ as in \eq{depth-function}. The evaluation oracle for $h_p(x)$ for any $x \in B(0,r/2)$ can be implemented to precision $\epsilon = 7\kappa\delta $ using $\log({1}/{\epsilon})$ queries to the membership oracle for $K$\; \label{lin:mem-to-sep-6}
    Compute $\tilde{g} = \algname{QuantumSubgradient}(h_p,\epsilon,L,0,n\epsilon^{1/2})$\;
    \textbf{Output:} the halfspace $ \lbrace x \in \R^n \mid {\left(30000R + 25\right) n^3\epsilon^{1/6}\kappa^2}/{\rho} \ge \langle \tilde{g}, x - p \rangle \rbrace $.
  }
\end{algorithm}

In this subsection we show how the approximate subgradient procedure (\algo{classical-smoothing}) fits into the reduction from separation to membership presented in \cite{lee2017efficient}. We use the \emph{height function} $h_p\colon \R^n \to \R$ defined in \cite{lee2017efficient} for any vector $p \in \R^n$, as
\begin{equation}
  \label{eq:depth-function}
  h_p(x) = -\max \{t \in \R \mid x + t\hat{p} \in K \},
\end{equation}
where $\hat{p}$ is the unit vector in the direction of $p$. The height function has the following properties:
\begin{proposition}[Lemmas 11 and 12 of \cite{lee2017efficient}]
  \label{prop:height-function-lipschitz}
  Let $K \subset \R^n$ be a convex set with $B_2(0,r) \subseteq K \subseteq B_2(0,R)$ for some $R>r>0$. Then for any $p \in \R^n$, the height function \eq{depth-function} satisfies
  \begin{enumerate}[nosep,label={\normalfont{(\roman*)}}]
  \item $h_p(x)$ is convex,
  \item $h_p(x) \le 0$ for all $x \in K$, and
  \item for all $\delta>0$, $h_p(x)$ is $\frac{R+\delta}{r - \delta}$-Lipschitz continuous for $x \in B_2(0,\delta)$.
  \end{enumerate}
\end{proposition}

Now we are ready to analyze \algo{halfspace}.

\begin{theorem}
  \label{thm:gen-halfspace}
  Let $K \subset \R^n$ be a convex set such that $B_2(0,r) \subseteq K \subseteq B_2(0,R)$ for some $R>r>0$. Let $\rho \in (0,1), \kappa = R/r$ and $\delta \in (0,\min\{r/7\kappa,1/7\kappa\})$. Then with probability at least $1 - \rho$, \algname{SeparatingHalfspace}$(K,p,\rho,\delta)$ outputs a halfspace that contains $K$ and not $p$.
\end{theorem}
\begin{proof}
Since $\delta \le \min\{r/7\kappa,1/7\kappa\}$, $\epsilon \le \min \lbrace 1,r \rbrace$. If $p \in B_2(K,\delta)$, the algorithm is trivially correct. If $p \notin B_2(0,R)$, the algorithm outputs a halfspace that contains $B_2(0,R)$ (and therefore contains $K$), and not $p$.

Finally, suppose $p \notin B_2(K,-\delta)$ and $p \in B_2(0,R)$. Since $\epsilon \ge \delta$, $p \notin B_2(K,-\epsilon)$.  The height function $h_p(x)$ is $3\kappa$-Lipschitz for all $x \in B_2(0,{r}/{2})$, where $\kappa:=R/r$. Define $r_1 = n\epsilon^{1/2}$. Since $\epsilon < \min \lbrace 1, r_1/n \rbrace$, \thm{general-quantum-gradient} implies
  \begin{equation}
    \label{eq:height-subgrad}
    h_p(x) \ge h_p(0) + \langle \tilde{g}, x \rangle - \zeta\norm{x}_\infty - 12nr_1\kappa
  \end{equation}
for any $x \in K$, where $\E\zeta \le \frac{15000\kappa n^{4/3}\epsilon^{1/3}}{r_1^{1/3}}$.

 Notice that $-{p}/{\kappa} \in K$ and $h_p\left( -{p}/{\kappa} \right) = h_p(0) - \frac{1}{\kappa}\norm{p}_2$. From \eq{height-subgrad},
  \begin{align}
    \label{eq:size-of-gradient}
    h_p(0) - \frac{1}{\kappa}\norm{p}_2 &\ge h_p(0) + \langle \tilde{g}, -{p}/{\kappa} \rangle - \frac{1}{\kappa}\zeta\norm{p}_\infty - 12nr_1\kappa,
  \end{align}
hence
\begin{align}
  \label{eq:size-of-gradient-2}
    \langle \tilde{g},p \rangle &\ge \norm{p}_2 - \zeta\norm{p}_\infty - 12nr_1\kappa^2.
  \end{align}
As claimed in \lin{mem-to-sep-6} of \algo{halfspace}, $h_p(x)$ can be evaluated with any precision $\epsilon$ such that $7\kappa\delta \le \epsilon$ using $O(\log(1/\epsilon))$ queries to a membership oracle with error $\delta$; the proof is deferred to \lem{height-eval}.

Since the membership oracle returns a negative response $p \notin B_2(K,-\delta)$, the error $\epsilon$ in $h_p(x)$ must be $\ge \delta$, and hence $p \notin B_2(K,-\epsilon)$. We are also given that $B_2(0,r) \subseteq K$. As a result, we have $\left( 1 - \frac{\epsilon}{r} \right)K \subseteq B_2(K,-\epsilon)$. Thus,
  \begin{equation}
    \label{eq:depth-of-subgrad}
    h_p(0) \ge - \left( 1 - \frac{\epsilon}{r} \right) \norm{p}_2 \ge -\norm{p}_2 + \epsilon\kappa.
  \end{equation}
  From \eq{height-subgrad}, \eq{size-of-gradient}, and \eq{depth-of-subgrad}, we have
  \begin{align}
    \label{eq:halfspace-cond}
    h_p(x) & \ge \langle \tilde{g},x-p \rangle - \zeta\norm{x}_\infty - \zeta\norm{p}_\infty - 12nr_1\kappa - 12nr_1\kappa^2 - \epsilon\kappa \\
           & \ge \langle \tilde{g},x-p \rangle - 2\zeta R - 24nr_1\kappa^2 - \epsilon\kappa,
  \end{align}
  so $\langle \tilde{g}, x-p \rangle \le \tilde{\zeta}$
  for all $x \in K$, where
  \begin{align}
    \label{eq:expected-error}
    \E\tilde{\zeta} &\le \frac{30000R n^{5/3}\epsilon^{1/3}\kappa}{r_1^{1/3}} +  24nr_1\kappa^2 + \epsilon\kappa \\
                    &\le 30000R n\epsilon^{1/6}\kappa +  24n^3\epsilon^{1/2}\kappa^2 + \epsilon\kappa \\
                    &\le (30000R + 25) n^3\epsilon^{1/6}\kappa^2.
  \end{align}
Thus the result follows from Markov's inequality.
\end{proof}

\begin{theorem}
  \label{thm:mem-to-sep}
  Let $K \subset \R^n$ be a convex set with $B_2(0,r) \subseteq K \subseteq B_2(0,R)$ and  $\kappa = R/r$ for some $R>r>0$, and let $\eta > 0$ be fixed. Further suppose that $R,r,\kappa=\poly(n)$. Then a separating oracle for $K$ with error $\eta$ can be implemented using $\softoh{1}$ queries to a membership oracle for $K$ and $\softoh{n}$ gates.
\end{theorem}
\begin{proof}
 Clearly, the unit vector in the direction $\tilde{g}$ (from \algo{halfspace}) determines a separating hyperplane given a point $p \notin B_2(K,-\epsilon)$. From \eq{size-of-gradient-2}, we have
  \begin{align}
    \langle \tilde{g},p \rangle \ge \norm{p}_2 -\left(\frac{15000n^{5/3}\epsilon^{1/3}\kappa}{r_1}\right)\norm{p}_\infty.
  \end{align}
  Letting $\frac{15000n^{5/3}\epsilon^{1/3}}{r_1} < \frac{1}{2\kappa^2}$, we have
  \begin{align}
    \label{eq:size-of-g}
    \norm{\tilde{g}}_2 R \ge r - \frac{R}{2\kappa}\ \hence \ \norm{\tilde{g}}_2 \ge \frac{1}{2\kappa}.
  \end{align}
  Thus, we have a separating oracle with error margin $\left(60000R + 50\right) n^{3}\epsilon^{1/6}\kappa^{3} \rho^{-1}$ and failure probability $\rho$. Setting $\rho = \bigl( \left(60000R + 50\right) n^{3}\epsilon^{1/6}\kappa^{3} \bigr)^{1/2}$, we have a composite error of $\bigl(\left(60000R + 50\right) n^{3}\epsilon^{1/6}\kappa^{3}\bigr)^{1/2}$.  To have error at most $\eta$, we take $\epsilon \le \eta^6/\bigl( ( 60000R + 15)^6 n^{18}\kappa^{18} \bigr)$.

  We finally obtain
  \begin{align}
    \label{eq:bound-on-epsilon}
     \delta = \frac{\epsilon}{7\kappa} \le \frac{1}{7\kappa}\min \Biggl\{ \frac{\eta^6}{\left( 60000R + 50\right)^6 n^{18}\kappa^{18}}, \frac{1}{8\kappa^6\Bigl(\frac{15000n^{5/3}\epsilon^{1/3}}{r_1}\Bigr)^3}, r, 1 \Biggr\}.
  \end{align}
  Consequently, we have $\SEP_\eta = \softoh{1} \MEM_\delta$, where
\begin{align}
     \delta = \frac{\epsilon}{7\kappa} \le \frac{1}{7\kappa}\min \Biggl\{ \frac{\eta^6}{\left( 60000R + 50\right)^6 n^{18}\kappa^{18}}, \frac{1}{8\kappa^6\Bigl(\frac{15000n^{5/3}\epsilon^{1/3}}{r_1}\Bigr)^3}, r, 1 \Biggr\}.
\end{align}
Therefore, ${1}/{\epsilon}$ and ${1}/{\delta}$ are both $O(\poly(n))$. Implementing the evaluation oracle takes $\poly(\log(1/\epsilon))$ membership queries and a further $\softoh{1}$ queries are used for the subgradient.

The evaluation requires $\softoh{1/\epsilon}$ gates and the $\algname{SmoothQuantumGradient}$ uses $n\poly(\log(1/\epsilon))$ gates. Thus a total of $\poly(\log({1}/{\eta}))$ queries and $n \poly(\log({1}/{\eta}))$ gates are used.
\end{proof}

\subsection{Separation to optimization}
\label{sec:sep-to-opt}

It is known that an optimization oracle for a convex set can be implemented in $\softoh{n}$ queries to a separation oracle. Specifically, Theorem 15 of \cite{lee2017efficient} states:

\begin{theorem}[Separation to Optimization]
  \label{thm:sep-to-opt}
  Let $K$ be a convex set satisfying $B_2(0,r) \subset K \subset B_2(0,R)$ and let $\kappa = 1/r$. For any $0 < \epsilon < 1$, with probability $1 - \epsilon$, we can compute $x \in B_2(K,\epsilon)$ such that $c^Tx \le \min_{x \in K} c^Tx + \epsilon \norm{c}_2$, using $O(n\log({n\kappa}/{\epsilon}))$ queries to $\SEP_\eta(K)$, where $\eta = \poly(\epsilon/n\kappa)$, and $\tilde{O}(n^3)$ arithmetic operations.
\end{theorem}

From \thm{sep-to-opt} and \thm{mem-to-sep}, we have the following result

\begin{theorem}[Membership to Optimization]
  \label{thm:mem-to-opt}
  Let $K$ be a convex set satisfying $B_2(0,r) \subset K \subset B_2(0,R)$ and let $\kappa = 1/r$. For any $0 < \epsilon < 1$, with probability $1 - \epsilon$, we can compute $x \in B_2(K,\epsilon)$ such that $c^Tx \le \min_{x \in K} c^Tx + \epsilon$, using $\softoh{n}$ queries to a membership oracle for $K$ with error $\delta$, where $\delta = O(\poly(\epsilon)) $, and $\tilde{O}(n^3)$ gates.
\end{theorem}

\begin{proof}
  Using \thm{mem-to-sep} with $\eta = \poly(\epsilon/n\kappa)$, each query to the separation oracle requires $\softoh{1}$ queries to a membership oracle with error $\delta = O(\poly(\epsilon))$. We make $\softoh{n}$ separation queries and perform a further $\tilde{O}(n^3)$ arithmetic operations, so the result follows.
\end{proof}

\thm{upper-bound-informal} follows directly from \thm{mem-to-opt}.


\section{Lower bound}\label{sec:lower-bound}

In this section, we prove our quantum lower bound on convex optimization (\thm{lower-main}). We prove separate lower bounds on membership queries (\sec{lower-membership}) and evaluation queries (\sec{lower-evaluation}). We then combine these lower bounds into a single optimization problem in \sec{main-together}, establishing \thm{lower-main}.

\subsection{Membership queries}\label{sec:lower-membership}
In this subsection, we establish a membership query lower bound using a reduction from the following search-with-wildcards problem:
\begin{theorem}[{\cite[Theorem 1]{ambainis2014quantum}}]\label{thm:wildcard}
For any $s\in\{0,1\}^{n}$, let $O_s$ be a wildcard oracle satisfying
\begin{align}\label{eqn:wildcard-defn}
O_{s}|T\>|y\>|0\>=|T\>|y\>|Q_{s}(T,y)\>
\end{align}
for all $T\subseteq\range{n}$ and $y\in\{0,1\}^{|T|}$, where $Q_{s}(T,y)=\delta[s_{|T}=y]$. Then the bounded-error quantum query complexity of determining $s$ is $O(\sqrt{n}\log n)$ and $\Omega(\sqrt{n})$.
\end{theorem}

We use \thm{wildcard} to give an $\Omega(\sqrt{n})$ lower bound on membership queries for convex optimization. Specifically, we consider the following \emph{sum-of-coordinates optimization problem}:
\begin{definition}\label{defn:sum-coordinate}
Let
\begin{align}\label{eqn:sum-coordinate}
\mathcal{C}_{s}:=\bigtimes_{i=1}^{n}[s_{i}-2,s_{i}+1],\qquad s_{i}\in\{0,1\}\ \ \forall\,i\in\range{n},
\end{align}
where $\bigtimes$ is the Cartesian product on different coordinates. In the \emph{sum-of-coordinates optimization} problem, the goal is to minimize
\begin{align}\label{eqn:sum-coordinate-function}
f(x)=\sum_{i\in\range{n}}x_{i}\quad \text{s.t. }x\in\mathcal{C}_{s}.
\end{align}
\end{definition}
\noindent
Intuitively, \defn{sum-coordinate} concerns an optimization problem on a hypercube where the function is simply the sum of the coordinates, but the position of the hypercube is unknown. Note that the function $f$ in \eqn{sum-coordinate-function} is convex and 1-Lipschitz continuous.

We prove the hardness of solving sum-of-coordinates optimization using its membership oracle:

\begin{theorem}\label{thm:main-membership}
Given an instance of the sum-of-coordinates optimization problem with membership oracle $O_{\mathcal{C}_{s}}$, it takes $\Omega(\sqrt{n})$ quantum queries to $O_{\mathcal{C}_{s}}$ to output an $\tilde{x}\in\mathcal{C}_{s}$ such that
\begin{align}\label{eqn:sum-of-coordinate}
f(\tilde{x})\leq\min_{x\in\mathcal{C}_{s}} f(x)+\frac{1}{3},
\end{align}
with success probability at least $0.9$.
\end{theorem}

\begin{proof}
Assume that we are given an arbitrary string $s\in\{0,1\}^{n}$ together with the membership oracle $O_{\mathcal{C}_{s}}$ for the sum-of-coordinates optimization problem.

We prove that a quantum query to $O_{\mathcal{C}_{s}}$ can be simulated by a quantum query to the oracle $O_{s}$ in \eqn{wildcard-defn} for search with wildcards. Consider an arbitrary point $x\in\R^{n}$ in the sum-of-coordinates problem. We partition $\range{n}$ into four sets:
\begin{align}
T_{x,0}&:=\big\{i\in\range{n}\mid x_{i}\in[-2,-1)\big\} \\
T_{x,1}&:=\big\{i\in\range{n}\mid x_{i}\in(1,2]\big\} \\
T_{x,\text{mid}}&:=\big\{i\in\range{n}\mid x_{i}\in[-1,1]\big\} \\
T_{x,\text{out}}&:=\big\{i\in\range{n}\mid |x_{i}|>2\big\},
\end{align}
and denote $T_{x}:=T_{x,0}\cup T_{x,1}$ and $y^{(x)}\in\{0,1\}^{|T_{x}|}$ such that
\begin{align}\label{eqn:membership-yx}
y^{(x)}_{i}=\begin{cases}
    0 & \text{if $i\in T_{x,0}$} \\
    1 & \text{if $i\in T_{x,1}$}.
  \end{cases}
\end{align}
We prove that $O_{\mathcal{C}_{s}}(x)=Q_{s}(T_{x},y^{(x)})$ if $T_{x,\text{out}}=\emptyset$, and $O_{\mathcal{C}_{s}}(x)=0$ otherwise. On the one hand, if $O_{\mathcal{C}_{s}}(x)=1$, we have $x\in\mathcal{C}_{s}$. Because for all $i\in\range{n}$, $x_{i}\in[s_{i}-2,s_{i}+1]\subset[-2,2]$ for both $s_{i}=0$ and $s_{i}=1$, we must have $T_{x,\text{out}}=\emptyset$. Now consider any $i\in T_{x}$. If $i\in T_{x,0}$, then $x_{i}\in[-2,-1)$. Because $x_{i}\in[0-2,0+1]$ and $x_{i}\notin[1-2,1+1]$, we must have $s_{i}=0$ since $x_{i}\in[s_{i}-2,s_{i}+1]$. Similarly, if $i\in T_{x,1}$, then we must have $s_{i}=1$. As a result of \eqn{membership-yx}, for all $i\in T_x$ we have $s_{i}=y^{(x)}_{i}$; in other words, $s_{|T_{x}}=y^{(x)}$ and $Q_{s}(T_{x},y^{(x)})=1=O_{\mathcal{C}_{s}}(x)$.

On the other hand, if $O_{\mathcal{C}_{s}}(x)=0$, there exists an $i_{0}\in\range{n}$ such that $x_{i_{0}}\notin[s_{i_{0}}-2,s_{i_{0}}+1]$. We must have $i_{0}\notin T_{x,\text{mid}}$ because $[-1,1]\subset[s_{i_{0}}-2,s_{i_{0}}+1]$ regardless of whether $s_{i_{0}}=0$ or $s_{i_{0}}=1$. Next, if $i_{0}\in T_{x,\text{out}}$, then $T_{x,\text{out}}\neq\emptyset$ and we correctly obtain $O_{\mathcal{C}_{s}}(x)=0$. The remaining cases are $i_{0}\in T_{x,0}$ and $i_{0}\in T_{x,1}$. If $i_{0}\in T_{x,0}$, because $x_{i_{0}}\in[-2,-1)\subset[0-2,0+1]$ and $x_{i_{0}}\notin[s_{i_{0}}-2,s_{i_{0}}+1]$, we must have $s_{i_{0}}=1$, and thus $s_{|T_{x}}\neq y^{(x)}$ because $y^{(x)}_{i_{0}}=0$ by \eqn{membership-yx}. If $i_{0}\in T_{x,1}$, we similarly have $s_{i_{0}}=0$, $y^{(x)}_{i_{0}}=1$, and thus $s_{|T_{x}}\neq y^{(x)}$. In both cases, $s_{|T_{x}}\neq y^{(x)}$, so $Q_{s}(T_{x},y^{(x)})=0=O_{\mathcal{C}_{s}}(x)$.

Therefore, we have established that $O_{\mathcal{C}_{s}}(x)=Q_{s}(T_{x},y^{(x)})$ if $T_{x,\text{out}}=\emptyset$, and $O_{\mathcal{C}_{s}}(x)=0$ otherwise. In other words, a quantum query to $O_{\mathcal{C}_{s}}$ can be simulated by a quantum query to $O_{s}$.

We next prove that a solution $\tilde{x}$ of the sum-of-coordinates problem satisfying \eqn{sum-of-coordinate} solves the search-with-wildcards problem in \thm{wildcard}. Because $\min_{x\in\mathcal{C}_{s}}f(x)=\sum_{i=1}^{n}(s_{i}-2)$, we have
\begin{align}\label{eqn:sum-of-coordinate-approx}
f(\tilde{x})=\sum_{i=1}^{n}\tilde{x}_{i}\leq\frac{1}{3}+\sum_{i=1}^{n}(s_{i}-2).
\end{align}
On the one hand, for all $j\in\range{n}$ we have $\tilde{x}_{j}\geq s_{j}-2$ since $\tilde{x}\in\mathcal{C}_{s}$; on the other hand, by \eqn{sum-of-coordinate-approx} we have
\begin{align}
\frac{1}{3}+\sum_{i=1}^{n}(s_{i}-2)\geq\sum_{i=1}^{n}\tilde{x}_{i}\geq \tilde{x}_{j}+\sum_{i\in\range{n},\ i\neq j}(s_{i}-2),
\end{align}
which implies $\tilde{x}_{j}\leq s_{j}-2+\frac{1}{3}$. In all,
\begin{align}\label{eqn:membership-tilde-x}
\tilde{x}_{i}\in[s_{i}-2,s_{i}-2+\tfrac{1}{3}]\quad\forall\,i\in\range{n}.
\end{align}

Define a rounding function $\sgn_{-3/2}\colon\R\to\{0,1\}$ as
\begin{align}\label{eqn:sgn-3/2}
\sgn_{-3/2}(z)=\begin{cases}
    0 & \text{if $z<-3/2$} \\
    1 & \text{otherwise}.
  \end{cases}
\end{align}
We prove that $\sgn_{-3/2}(\tilde{x})=s$ (here $\sgn_{-3/2}$ is applied on all $n$ coordinates, respectively). For all $i\in\range{n}$, if $s_{i}=0$, then $\tilde{x}_{i}\in[-2,-\frac{5}{3}]\subset(-\infty,-\frac{3}{2})$ by \eqn{membership-tilde-x}, which implies $\sgn_{-3/2}(\tilde{x}_{i})=0$ by \eqn{sgn-3/2}. Similarly, if $s_{i}=1$, then $\tilde{x}_{i}\in[-1,-\frac{2}{3}]\subset(-\frac{3}{2},+\infty)$ by \eqn{membership-tilde-x}, which implies $\sgn_{-3/2}(\tilde{x}_{i})=1$ by \eqn{sgn-3/2}.

In all, if we can solve the sum-of-coordinates optimization problem with an $\tilde{x}$ satisfying \eqn{sum-of-coordinate}, we can solve the search-with-wildcards problem. By \thm{main-membership}, the search-with-wildcards problem has quantum query complexity $\Omega(\sqrt{n})$; since a query to the membership oracle $O_{\mathcal{C}_{s}}$ can be simulated by a query to the wildcard oracle $O_{s}$, we have established an $\Omega(\sqrt{n})$ quantum lower bound on membership queries to solve the sum-of-coordinates optimization problem.
\end{proof}

\subsection{Evaluation queries}\label{sec:lower-evaluation}
In this subsection, we establish an evaluation query lower bound by considering the following \emph{max-norm optimization problem}:

\begin{definition}\label{defn:max-OPT}
In the \emph{max-norm optimization problem}, the goal is to minimize a function $f_{c}\colon\R^{n}\rightarrow\R$ satisfying
\begin{align}\label{eqn:max-OPT-defn}
f_{c}(x)=\max_{i\in\range{n}}|\pi(x_{i})-c_{i}|+\Big(\sum_{i=1}^{n}|\pi(x_{i})-x_{i}|\Big)
\end{align}
for some $c\in\{0,1\}^{n}$, where $\pi\colon\R\rightarrow[0,1]$ is defined as
\begin{align}\label{eqn:projection-definition-main}
\pi(x)=\begin{cases}
    0 & \text{if $x<0$} \\
    x & \text{if $0\leq x\leq 1$} \\
    1 & \text{if $x>1$}.
  \end{cases}
\end{align}
\end{definition}
\noindent
Observe that for all $x \in [0,1]^n$, we have $f_c(x) = \max_{i \in \range{n}} |x_i - c_i|$.
Intuitively, \defn{max-OPT} concerns an optimization problem under the max-norm (i.e., $L_{\infty}$ norm) distance from $c$ for all $x$ in the unit hypercube $[0,1]^{n}$; for all $x$ not in the unit hypercube, the optimizing function pays a penalty of the $L_{1}$ distance between $x$ and its projection $\pi(x)$ onto the unit hypercube. The function $f_{c}$ is 2-Lipschitz continuous with a unique minimum at $x=c$; we prove in \lem{evaluation-convexity} that $f_{c}$ is convex.

We prove the hardness of solving max-norm optimization using its evaluation oracle:

\begin{theorem}\label{thm:main-evaluation}
Given an instance of the max-norm optimization problem with an evaluation oracle $O_{f_{c}}$, it takes $\Omega(\sqrt{n}/\log n)$ quantum queries to $O_{f_{c}}$ to output an $\tilde{x}\in[0,1]^{n}$ such that
\begin{align}\label{eqn:max-norm-main}
f_{c}(\tilde{x})\leq\min_{x\in[0,1]^{n}}f_{c}(x)+\frac{1}{3},
\end{align}
with success probability at least $0.9$.
\end{theorem}

The proof of \thm{main-evaluation} has two steps. First, we prove a weaker lower bound with respect to the precision of the evaluation oracle:

\begin{lemma}\label{lem:main-evaluation-weak}
Suppose we are given an instance of the max-norm optimization problem with an evaluation oracle $O_{f_{c}}$ that has precision $0<\delta<0.05$, i.e., $f_{c}$ is provided with $\lceil\log_{2}(1/\delta)\rceil$ bits of precision. Then it takes $\Omega(\sqrt{n}/\log (1/\delta))$ quantum queries to $O_{f_{c}}$ to output an $\tilde{x}\in[0,1]^{n}$ such that
\begin{align}\label{eqn:max-norm}
f_{c}(\tilde{x})\leq\min_{x\in[0,1]^{n}}f_{c}(x)+\frac{1}{3},
\end{align}
with success probability at least $0.9$.
\end{lemma}

The second step simulates a perfectly precise query to $f_{c}$ by a rough query:
\begin{lemma}\label{lem:precision-simulation}
One classical (resp., quantum) query to $O_{f_{c}}$ with perfect precision can be simulated by one classical query (resp., two quantum queries) to $O_{f_{c}}$ with precision $1/5n$.
\end{lemma}

\thm{main-evaluation} simply follows from the two propositions above: by \lem{precision-simulation}, we can assume that the evaluation oracle $O_{f_{c}}$ has precision $1/5n$, so \lem{main-evaluation-weak} implies that it takes $\Omega(\sqrt{n}/\log 5n)=\Omega(\sqrt{n}/\log n)$ quantum queries to $O_{f_{c}}$ to output an $\tilde{x}\in[0,1]^{n}$ satisfying \eqn{max-norm-main} with success probability 0.9.

The proofs of \lem{main-evaluation-weak} and \lem{precision-simulation} are given in \sec{evaluation-low-precision} and \sec{discretization}, respectively.

\subsubsection{$\tilde{\Omega}(\sqrt{n})$ quantum lower bound on a low-precision evaluation oracle}\label{sec:evaluation-low-precision}
Similar to the proof of \thm{main-membership}, we also use \thm{wildcard} (the quantum lower bound on search with wildcards) to give a quantum lower bound on the number of evaluation queries required to solve the max-norm optimization problem.

\begin{proof}[Proof of \lem{main-evaluation-weak}]
Assume that we are given an arbitrary string $c\in\{0,1\}^{n}$ together with the evaluation oracle $O_{f_{c}}$ for the max-norm optimization problem. To show the lower bound, we reduce the search-with-wildcards problem to the max-norm optimization problem.

We first establish that an evaluation query to $O_f$ can be simulated using wildcard queries on $c$. Notice that if we query an arbitrary $x\in\R^{n}$, by \eqn{max-OPT-defn} we have
\begin{align}\label{eqn:projection-relationship}
f_{c}(x)=\max_{i\in\range{n}}|\pi(x_{i})-c_{i}|+\Big(\sum_{i=1}^{n}|\pi(x_{i})-x_{i}|\Big)=f_{c}(\pi(x))+\Big(\sum_{i=1}^{n}|\pi(x_{i})-x_{i}|\Big)
\end{align}
where $\pi(x):=(\pi(x_{1}),\ldots,\pi(x_{n}))$. In particular, the difference of $f_{c}(x)$ and $f_{c}(\pi(x))$ is an explicit function of $x$ that is independent of $c$. Thus the query $O_{f_{c}}(x)$ can be simulated using one query to $O_{f_{c}}(\pi(x))$ where $\pi(x)\in[0,1]^{n}$. It follows that we can restrict ourselves without loss of generality to implementing evaluation queries for $x \in [0,1]^n$.

Now we consider a decision version of oracle queries to $f_{c}$, denoted $O_{f_{c,\text{dec}}}$, where the function $f_{c,\text{dec}}\colon[0,1]^{n}\times[0,1]\to\{0,1\}$ satisfies
\begin{align}
f_{c,\text{dec}}(x,t)=\delta[f_{c}(x)\leq t].
\end{align}
(We restrict to $t\in[0,1]$ because $f_{c}(x)\in[0,1]$ always holds for $x\in[0,1]^{n}$.) Using binary search, a query to $O_{f_{c}}$ with precision $\delta$ can be simulated by at most $\lceil\log_{2}(1/\delta)\rceil=O(\log 1/\delta)$ queries to the oracle $O_{f_{c,\text{dec}}}$.

Next, we prove that a query to $O_{f_{c,\text{dec}}}$ can be simulated by a query to the search-with-wildcards oracle $O_{c}$ in \eqn{wildcard-defn}. Consider an arbitrary query $(x,t)\in[0,1]^{n}\times[0,1]$ to $O_{f_{c,\text{dec}}}$. For convenience, we denote $J_{0,t}:=[0,t]$, $J_{1,t}:=[1-t,1]$, and
\begin{align}
I_{0,t}&:=J_{0,t}-(J_{0,t}\cap J_{1,t}) \label{eqn:defn-I-0} \\
I_{1,t}&:=J_{1,t}-(J_{0,t}\cap J_{1,t}) \label{eqn:defn-I-1} \\
I_{\text{mid},t}&:=J_{0,t}\cap J_{1,t} \label{eqn:defn-I-mid} \\
I_{\text{out},t}&:=[0,1]-(J_{0,t}\cup J_{1,t}). \label{eqn:defn-I-out}
\end{align}

We partition $\range{n}$ into four sets:
\begin{align}
T_{x,0,t}&:=\big\{i\in\range{n}\mid x_{i}\in I_{0,t}\big\} \label{eqn:defn-Tx-0} \\
T_{x,1,t}&:=\big\{i\in\range{n}\mid x_{i}\in I_{1,t}\big\} \label{eqn:defn-Tx-1} \\
T_{x,\text{mid},t}&:=\big\{i\in\range{n}\mid x_{i}\in I_{\text{mid},t}\big\} \label{eqn:defn-Tx-mid} \\
T_{x,\text{out},t}&:=\big\{i\in\range{n}\mid x_{i}\in I_{\text{out},t}\big\}. \label{eqn:defn-Tx-out}
\end{align}
The strategy here is similar to the proof of \thm{main-membership}: $T_{x,\text{mid},t}$ corresponds to the coordinates such that $|x_{i}-c_{i}|\leq t$ regardless of whether $c_{i}=0$ or 1 (and hence $c_{i}$ does not influence whether or not $\max_{i\in\range{n}}|x_{i}-c_{i}|\leq t$); $T_{x,\text{out},t}$ corresponds to the coordinates such that $|x_{i}-c_{i}|>t$ regardless of whether $c_{i}=0$ or 1 (so $\max_{i\in\range{n}}|x_{i}-c_{i}|>t$ provided $T_{x,\text{out},t}$ is nonempty); and $T_{x,0,t}$ (resp., $T_{x,1,t}$) corresponds to the coordinates such that $|x_{i}-c_{i}|\leq t$ only when $c_{i}=0$ (resp., $c_{i}=1$).

Denote $T_{x,t}:=T_{x,0,t}\cup T_{x,1,t}$ and let $y^{(x,t)}\in\{0,1\}^{|T_{x,t}|}$ such that
\begin{align}\label{eqn:membership-yxt}
y^{(x,t)}_{i}=\begin{cases}
    0 & \text{if $i\in T_{x,0,t}$} \\
    1 & \text{if $i\in T_{x,1,t}$}.
  \end{cases}
\end{align}
We will prove that $O_{f_{c,\text{dec}}}(x)=Q_{c}(T_{x,t},y^{(x,t)})$ if $T_{x,\text{out},t}=\emptyset$, and $O_{f_{c,\text{dec}}}(x)=0$ otherwise.

On the one hand, if $O_{f_{c,\text{dec}}}(x)=1$, we have $f_{c}(x)\leq t$. In other words, for all $i\in\range{n}$ we have $|x_{i}-c_{i}|\leq t$, which implies
\begin{align}\label{eqn:evaluation-decision-1}
x_{i}\in J_{c_{i},t}\quad\forall\,i\in\range{n}.
\end{align}
Since $J_{c_{i},t}\subseteq J_{0,t}\cup J_{1,t}$, we have $x_{i}\in J_{0,t}\cup J_{1,t}$ for all $i\in\range{n}$, and thus $T_{x,\text{out},t}=\emptyset$ by \eqn{defn-I-out} and \eqn{defn-Tx-out}. Now consider any $i\in T_{x,t}$. If $i\in T_{x,0,t}$, then $x_{i}\in I_{0,t}$ by \eqn{defn-Tx-0}. By \eqn{defn-I-0} we have $x_{i}\in J_{0,t}$ and $x_{i}\notin J_{1,t}$, and thus $c_{i}=0$ by \eqn{evaluation-decision-1}. Similarly, if $i\in T_{x,1,t}$, then we must have $c_{i}=1$. As a result of \eqn{membership-yxt}, for all $i\in T_{x,t}$ we have $c_{i}=y^{(x,t)}_{i}$; in other words, $c_{|T_{x,t}}=y^{(x,t)}$ and $Q_{c}(T_{x,t},y^{(x,t)})=1=O_{f_{c,\text{dec}}}(x)$.

On the other hand, if $O_{f_{c,\text{dec}}}(x)=0$, there exists an $i_{0}\in\range{n}$ such that
\begin{align}\label{eqn:evaluation-decision-2}
x_{i_{0}}\notin J_{c_{i_{0}},t}.
\end{align}
Therefore, we must have $i_{0}\notin T_{x,\text{mid},t}$ since \eqn{defn-I-mid} implies $I_{\text{mid},t}=J_{0,t}\cap J_{1,t}\subseteq J_{c_{i_{0}},t}$. Next, if $i_{0}\in T_{x,\text{out},t}$, then $T_{x,\text{out},t}\neq\emptyset$ and we correctly obtain $O_{f_{c,\text{dec}}}(x)=0$. The remaining cases are $i_{0}\in T_{x,0,t}$ and $i_{0}\in T_{x,1,t}$.

If $i_{0}\in T_{x,0,t}$, then $y^{(x,t)}_{i_{0}}=0$ by \eqn{membership-yxt}. By \eqn{defn-Tx-0} we have $x_{i_{0}}\in I_{0,t}$, and by \eqn{defn-I-0} we have $x_{i_{0},t}\in J_{0,t}$ and $x_{i_{0}}\notin J_{1,t}$; therefore, we must have $c_{i_{0}}=1$ by \eqn{evaluation-decision-2}. As a result, $c_{|T_{x,t}}\neq y^{(x,t)}$ at $i_{0}$. If $i_{0}\in T_{x,1,t}$, we similarly have $c_{i_{0}}=0$, $y^{(x,t)}_{i_{0}}=1$, and thus $c_{|T_{x,t}}\neq y^{(x,t)}$ at $i_{0}$. In either case, $c_{|T_{x,t}}\neq y^{(x,t)}$, and $Q_{c}(T_{x,t},y^{(x,t)})=0=O_{f_{c,\text{dec}}}(x)$.

Therefore, we have established that $O_{f_{c,\text{dec}}}(x)=Q_{c}(T_{x,t},y^{(x,t)})$ if $T_{x,\text{out},t}=\emptyset$, and $O_{f_{c,\text{dec}}}(x)=0$ otherwise. In other words, a quantum query to $O_{f_{c,\text{dec}}}$ can be simulated by a quantum query to the search-with-wildcards oracle $O_{c}$. Together with the fact that a query to $O_{f_{c}}$ with precision $\delta$ can be simulated by $O(\log 1/\delta)$ queries to $O_{f_{c,\text{dec}}}$, it can also be simulated by $O(\log 1/\delta)$ queries to $O_{c}$.

We next prove that a solution $\tilde{x}$ of the max-norm optimization problem satisfying \eqn{max-norm} solves the search-with-wildcards problem in \thm{wildcard}. Because $\min_{x\in[0,1]^{n}}f_{c}(x)=0$, considering the precision of at most $\delta<0.05$ we have
\begin{align}\label{eqn:max-norm-approx}
f_{c}(\tilde{x})\leq\tfrac{1}{3}+\delta\leq 0.4.
\end{align}
In other words,
\begin{align}\label{eqn:max-norm-tilde-x}
\tilde{x}_{i}\in[c_{i}-0.4,c_{i}+0.4]\quad\forall\,i\in\range{n}.
\end{align}

Similar to \eqn{sgn-3/2}, we define a rounding function $\sgn_{1/2}\colon\R\to\{0,1\}$ as
\begin{align}\label{eqn:sgn1/2}
\sgn_{1/2}(z)=\begin{cases}
    0 & \text{if $z<1/2$} \\
    1 & \text{otherwise}.
  \end{cases}
\end{align}
We prove that $\sgn_{1/2}(\tilde{x})=c$ (here $\sgn_{1/2}$ is applied coordinate-wise). For all $i\in\range{n}$, if $c_{i}=0$, then $\tilde{x}_{i}\in[0,0.4]\subset(-\infty,1/2)$ by \eqn{max-norm-tilde-x}, which implies $\sgn_{1/2}(\tilde{x}_{i})=0$ by \eqn{sgn1/2}. Similarly, if $c_{i}=1$, then $\tilde{x}_{i}\in[0.6,1]\subset(1/2,+\infty)$ by \eqn{max-norm-tilde-x}, which implies $\sgn_{1/2}(\tilde{x}_{i})=1$ by \eqn{sgn1/2}.

We have shown that if we can solve the max-norm optimization problem with an $\tilde{x}$ satisfying \eqn{max-norm}, we can solve the search-with-wildcards problem. By \thm{main-membership}, the search-with-wildcards problem has quantum query complexity $\Omega(\sqrt{n})$; since a query to the evaluation oracle $O_{f_{c}}$ can be simulated by $O(\log 1/\delta)$ queries to the wildcard oracle $O_{c}$, we have established an $\Omega(\sqrt{n}/\log (1/\delta))$ quantum lower bound on the number of evaluation queries needed to solve the max-norm optimization problem.
\end{proof}

\subsubsection{Discretization: simulating perfectly precise queries by low-precision queries}\label{sec:discretization}
In this subsection we prove \lem{precision-simulation}, which we rephrase more formally as follows. Throughout this subsection, the function $f_{c}$ in \eqn{max-OPT-defn} is abbreviated by $f$ for notational convenience.
\begin{lemma}\label{lem:precision-simulation-formal}
Assume that $\hat{f}\colon[0,1]^{n}\to [0,1]$ satisfies $|\hat{f}(x)-f(x)|\leq\tfrac{1}{5n}\ \forall\,x\in[0,1]^{n}$. Then one classical (resp., quantum) query to $O_{f}$ can be simulated by one classical query (resp., two quantum queries) to $O_{\hat{f}}$.
\end{lemma}

To achieve this, we present an approach that we call \emph{discretization}. Instead of considering queries on all of $[0,1]^{n}$, we only consider a discrete subset $D_{n} \subseteq [0,1]^{n}$ defined as
\begin{align}\label{eqn:discrete-set}
  D_n := \bigl\{ \chi(a,\pi) \mid a \in \{0,1\}^n \text{ and } \pi \in S_n \bigr\},
\end{align}
where $S_{n}$ is the symmetric group on $\range{n}$ and $\chi\colon \{0,1\}^n \times S_n \to [0,1]^n$ satisfies
\begin{align}
  \chi(a,\pi)_i = (1-a_{i})\tfrac{\pi(i)}{2n+1}+a_{i}(1-\tfrac{\pi(i)}{2n+1})
  \quad \forall\,i\in\range{n}.
  \label{eq:chi}
\end{align}
Observe that $D_{n}$ is a subset of $[0,1]^{n}$.

Since $|S_{n}|=n!$ and there are $2^{n}$ choices for $a\in\{0,1\}^{n}$, we have $|D_{n}|=2^{n}n!$. For example, when $n=2$, we have
\begin{align}
D_{2}=\big\{
  \big(\tfrac{1}{5},\tfrac{2}{5}\big),
  \big(\tfrac{1}{5},\tfrac{3}{5}\big),
  \big(\tfrac{4}{5},\tfrac{2}{5}\big),
  \big(\tfrac{4}{5},\tfrac{3}{5}\big),
  \big(\tfrac{2}{5},\tfrac{1}{5}\big),
  \big(\tfrac{2}{5},\tfrac{4}{5}\big),
  \big(\tfrac{3}{5},\tfrac{1}{5}\big),
  \big(\tfrac{3}{5},\tfrac{4}{5}\big)\big\}
\end{align}
with $|D_{2}|=2^{2}\cdot 2!=8$.

We denote the restriction of the oracle $O_{f}$ to $D_{n}$ by $O_{f|D_{n}}$, i.e.,
\begin{align}\label{eqn:evaluation-discrete}
O_{f|D_{n}}|x\>|0\>=|x\>|f(x)\>\qquad\forall\,x\in D_{n}.
\end{align}
In fact, this restricted oracle entirely captures the behavior of the unrestricted function.

\begin{lemma}[Discretization]\label{lem:discretization}
A classical (resp., quantum) query to $O_{f}$ can be simulated using one classical query (resp., two quantum queries) to $O_{f|D_{n}}$.
\end{lemma}

\begin{algorithm}[htbp]
\SetKwInput{KwInput}{Input}
\KwInput{$x\in[0,1]^{n}$\;}
\SetKwInput{KwOut}{Output}
\KwOut{$f(x)\in[0,1]$\;}
Compute $b\in\{0,1\}^{n}$ and $\sigma\in S_{n}$ such that the $2n$ numbers $x_{1},x_{2},\ldots,x_{n},1-x_{1},\ldots,1-x_{n}$ are arranged in decreasing order as \label{lin:discretization-1}
\begin{align}
&\hspace{-4mm} b_{\sigma(1)}x_{\sigma(1)}+(1-b_{\sigma(1)})(1-x_{\sigma(1)})\geq\cdots\geq b_{\sigma(n)}x_{\sigma(n)}+(1-b_{\sigma(n)})(1-x_{\sigma(n)}) \nonumber \\
&\hspace{-4mm} \geq (1-b_{\sigma(n)})x_{\sigma(n)}+b_{\sigma(n)}(1-x_{\sigma(n)})\geq\cdots\geq (1-b_{\sigma(1)})x_{\sigma(1)}+b_{\sigma(1)}(1-x_{\sigma(1)});\label{eqn:ord-x}
\end{align}

Compute $x^{*}\in D_{n}$ such that $\chi(b,\sigma^{-1}) = x^*$ (where $\chi$ is defined in \eq{chi})\; \label{lin:discretization-2}

Query $f(x^{*})$ and let $k^{*}=(2n+1)(1-f(x^{*}))$\; \label{lin:discretization-3}

 Return\label{lin:discretization-4}
 \begin{align}\label{eqn:discretization-function-value}
   f(x) = \begin{cases}
   (1-b_{\sigma(n)})x_{\sigma(n)}+b_{\sigma(n)}(1-x_{\sigma(n)})
   & \text{if $k^{*}=n+1$} \\
   b_{\sigma(k^{*})}x_{\sigma(k^{*})}+(1-b_{\sigma(k^{*})})(1-x_{\sigma(k^{*})})
   & \text{otherwise}.
   \end{cases}
 \end{align}
\caption{Simulate one query to $O_{f}$ using one query to $O_{f|D_{n}}$.}
\label{algo:discretization}
\end{algorithm}

We prove this proposition by giving an algorithm (\algo{discretization}) that performs the simulation. The main idea is to compute $f(x)$ only using $x$ and $f(x^{*})$ for some $x^{*}\in D_{n}$. We observe that max-norm optimization has the following property: if two strings $x\in[0,1]^{n}$ and $x^{*}\in D_{n}$ are such that $x_{1},\ldots,x_{n},1-x_{1},\ldots,1-x_{n}$ and $x^{*}_{1},\ldots,x^{*}_{n},1-x^{*}_{1},\ldots,1-x^{*}_{n}$ have the same ordering, then
\begin{align}
\arg\max_{i\in\range{n}}|x_{i}-c_{i}|=\arg\max_{i\in\range{n}}|x^{*}_{i}-c_{i}|.
\end{align}
Furthermore, $x^{*}\in D_{n}$ ensures that $\{x^{*}_{1},\ldots,x^{*}_{n},1-x^{*}_{1},\ldots,1-x^{*}_{n}\}=\{\frac{1}{2n+1},\ldots,\frac{2n}{2n+1}\}$ are $2n$ distinct numbers, so knowing the value of $f(x^{*})$ is sufficient to determine the value of the $\arg\max$ above (denoted $i^{*}$) and the corresponding $c_{i^{*}}$. We can then recover $f(x)=|x_{i^{*}}-c_{i^{*}}|$ using the given value of $x$. Moreover, $f(x^{*})$ is an integer multiple of $\frac{1}{2n+1}$; even if $f(x^{*})$ can only be computed with precision $\frac{1}{5n}$, we can round it to the closest integer multiple of $\frac{1}{2n+1}$ which is exactly $f(x^{*})$, since the distance $\frac{2n+1}{5n}<\frac{1}{2}$. As a result, we can precisely compute $f(x^{*})$ for all $x\in D_{n}$, and thus we can precisely compute $f(x)$.

We illustrate \algo{discretization} by a simple example. For convenience, we define an order function $\Ord \colon[0,1]^{n}\to \{0,1\}^{n}\times S_{n}$ by $\Ord(x)=(b,\sigma)$ for all $x\in [0,1]^{n}$, where $b$ and $\sigma$ satisfy Eq.\ \eqn{ord-x}.

\hd{An example with $n=3$.} Consider the case where the ordering in \eqn{ord-x} is
\begin{align}
1-x_{3}\geq x_{1}\geq x_{2}\geq 1-x_{2}\geq 1-x_{1}\geq x_{3}.
\label{eq:example_ord}
\end{align}
Then \algo{discretization} proceeds as follows:

\begin{itemize}
\item \lin{discretization-1}: With the ordering \eq{example_ord}, we have $\sigma(1)=3$, $\sigma(2)=1$, $\sigma(3)=2$; $b_{3}=0$, $b_{1}=1$, $b_{2}=1$.

\item \lin{discretization-2}: The point $x^{*}\in D_{3}$ that we query given $\Ord(x)$ satisfies $1-x^{*}_{3}=6/7$, $x^{*}_{1}=5/7$, $x^{*}_{2}=4/7$, $1-x^{*}_{2}=3/7$, $1-x^{*}_{1}=2/7$, and $x^{*}_{3}=1/7$; in other words, $x^{*}=(5/7,4/7,1/7)$.

\item \lin{discretization-3}: Now we query $f(x^{*})$. Since $f(x^{*})$ is a multiple of $1/7$ and $f(x^{*})\in[1/7,6/7]$, there are only 6 possibilities: $f(x^{*})=6/7$, $f(x^{*})=5/7$, $f(x^{*})=4/7$, $f(x^{*})=3/7$,  $f(x^{*})=2/7$, or $f(x^{*})=1/7$.

\item[] After running \lin{discretization-1}, \lin{discretization-2}, and \lin{discretization-3}, we have a point $x^{*}$ from the discrete set $D_{3}$ such that $\Ord(x)=\Ord(x^{*})$. Since they have the same ordering and $|x_{i}-c_{i}|$ is either $x_{i}$ or $1-x_{i}$ for all $i\in\range{3}$, the function value $f(x^{*})$ should essentially reflect the value of $f(x)$; this is made precise in \lin{discretization-4}.

\item \lin{discretization-4}: Depending on the value of $f(x^{*})$, we have six cases:
\begin{itemize}
\item $f(x^{*})=6/7$: In this case, we must have $c_{3}=1$, so that $|x_{3}-c_{3}|=|1/7-1|=6/7$ ($|x_{1}-c_{1}|$ can only give 5/7 or 2/7, and $|x_{2}-c_{2}|$ can only give 4/7 or 3/7). Because $1-x_{3}$ is the largest in \eq{example_ord}, we must have $f(x)=1-x_{3}$.

\item $f(x^{*})=5/7$: In this case, we must have $c_{1}=0$, so that $|x_{1}-c_{1}|=|5/7-0|=5/7$. Furthermore, we must have $c_{3}=1$ (otherwise if $c_{3}=0$, $f(x)\geq|x_{3}-c_{3}|=6/7$). As a result of \eq{example_ord}, we must have $f(x)=x_{1}$ since $x_{1}\geq x_{3}$ and $x_{1}\geq\max\{x_{2},1-x_{2}\}$.

\item $f(x^{*})=4/7$: In this case, we must have $c_{2}=0$, so that $|x_{2}-c_{2}|=|4/7-0|=4/7$. Furthermore, we must have $c_{3}=1$ (otherwise if $c_{3}=0$, $f(x)\geq|x_{3}-c_{3}|=6/7$) and $c_{1}=1$ (otherwise if $c_{1}=0$, $f(x)\geq|x_{1}-c_{1}|=5/7$). As a result of \eq{example_ord}, we must have $f(x)=x_{2}$ since $x_{2}\geq 1-x_{1}\geq 1-x_{3}$.

\item $f(x^{*})=3/7$: In this case, we must have $c_{2}=1$, so that $|x_{2}-c_{2}|=|4/7-1|=3/7$. Furthermore, we must have $c_{3}=1$ (otherwise if $c_{3}=0$, $f(x)\geq|x_{3}-c_{3}|=6/7$) and $c_{1}=1$ (otherwise if $c_{1}=0$, $f(x)\geq|x_{1}-c_{1}|=5/7$). As a result of \eq{example_ord}, we must have $f(x)=1-x_{2}$ since since $1-x_{2}\geq 1-x_{1}\geq 1-x_{3}$.

\item $f(x^{*})=2/7$ or $f(x^{*})=1/7$: This two cases are impossible because $f(x^{*})\geq |x_{2}-c_{2}|=|4/7-c_{2}|\geq 3/7$, no matter $c_{2}=0$ or $c_{2}=1$.
\end{itemize}
\end{itemize}

While \algo{discretization} is a classical algorithm for querying $O_f$ using a query to $O_{f|D_n}$, it is straightforward to perform this computation in  superposition using standard techniques to obtain a quantum query to $O_f$. However, note that this requires two queries to a quantum oracle for $O_{f|D_n}$ since we must uncompute $f(x^*)$ after computing $f(x)$.

Having the discretization technique at hand, \lem{precision-simulation-formal} is straightforward.
\begin{proof}[Proof of \lem{precision-simulation-formal}]
Recall that $|\hat{f}(x)-f(x)|\leq\tfrac{1}{5n}\ \forall\,x\in[0,1]^{n}$. We run \algo{discretization} to compute $f(x)$ for the queried value of $x$, except that in \lin{discretization-3} we take $k^{*}=\lceil(2n+1)(1-\hat{f}(x^{*}))\rfloor$ (here $\lceil a\rfloor$ is the closest integer to $a$). Because $|\hat{f}(x^{*})-f(x^{*})|\leq\tfrac{1}{5n}$, we have
\begin{align}
\big|(2n+1)(1-\hat{f}(x^{*}))-(2n+1)(1-f(x^{*}))\big|=(2n+1)|\hat{f}(x^{*})-f(x^{*})|\leq\tfrac{2n+1}{5n}<\tfrac{1}{2};
\end{align}
as a result, $k^{*}=(2n+1)(1-f(x^{*}))$ because the latter is an integer (see \lem{discrete-line-3}). Therefore, due to the correctness of \algo{discretization} established in \sec{discretization-appendix}, and noticing that the evaluation oracle is only called at \lin{discretization-3} (with the replacement described above), we successfully simulate one query to $O_{f}$ by one query to $O_{\hat{f}}$ (actually, to $O_{\hat{f}|D_{n}}$).
\end{proof}

The full analysis of \algo{discretization} is deferred to \sec{discretization-appendix}. In particular,
\begin{itemize}
\item In \sec{discrete-line-12} we prove that the discretized vector $x^{*}$ obtained in \lin{discretization-2} is a good approximation of $x$ in the sense that $\Ord(x^{*})=\Ord(x)$;
\item In \sec{discrete-line-3} we prove that the value $k^{*}$ obtained in \lin{discretization-3} satisfies $k^{*}\in\{1,\ldots,n+1\}$;
\item In \sec{discrete-line-4} we finally prove that the output returned in \lin{discretization-4} is correct.
\end{itemize}

\subsection{Proof of \thm{lower-main}}\label{sec:main-together}

We now prove \thm{lower-main} using \thm{main-membership} and \thm{main-evaluation}. Recall that our lower bounds on membership and evaluation queries are both proved on the $n$-dimensional hypercube. It remains to combine the two lower bounds to establish them simultaneously.

\begin{theorem}\label{thm:main-all}
Let $\mathcal{C}_{s}:=\bigtimes_{i=1}^{n}[s_{i}-2,s_{i}+1]$ for some $s\in\{0,1\}^{n}$. Consider a function $f\colon \mathcal{C}_{s}\times [0,1]^{n}\to\R$ such that $f(x)=f_{\emph{M}}(x)+f_{\emph{E},c}(x)$, where for any $x=(x_{1},x_{2},\ldots,x_{2n})\in\mathcal{C}_{s}\times[0,1]^{n}$,
\begin{align}
f_{\emph{M}}(x)=\sum_{i=1}^{n}x_{i},\qquad f_{\emph{E},c}(x)=\max_{i\in\{n+1,\ldots,2n\}}|x_{i}-c_{i-n}|
\end{align}
for some $c\in\{0,1\}^{n}$. Then outputting an $\tilde{x}\in\mathcal{C}_{s}\times[0,1]^{n}$ satisfying
\begin{align}\label{eqn:approx-all}
f(\tilde{x})\leq\min_{x\in\mathcal{C}_{s}\times[0,1]^{n}}f(x)+\tfrac{1}{3}
\end{align}
with probability at least $0.8$ requires $\Omega(\sqrt{n})$ quantum queries to $O_{\mathcal{C}_{s}\times[0,1]^{n}}$ and $\Omega(\sqrt{n}/\log n)$ quantum queries to $O_{f}$.
\end{theorem}

Notice that the dimension of the optimization problem above is $2n$ instead of $n$; however, the constant overhead of 2 does not influence the asymptotic lower bounds.

\begin{proof}[Proof of \thm{main-all}]
First, we prove that
\begin{align}
\min_{x\in\mathcal{C}_{s}\times[0,1]^{n}}f(x)=S\qquad\text{and}\qquad\arg\min_{x\in\mathcal{C}_{s}\times[0,1]^{n}}f(x)=(s-2_{n},c),
\end{align}
where $2_{n}$ is the $n$-dimensional all-twos vector and $S:=\sum_{i=1}^{n}(s_{i}-2)$. On the one hand,
\begin{align}\label{eqn:lower-all-m}
f_{\text{M}}(x)\geq S\qquad\forall\,x\in\mathcal{C}_{s}\times[0,1]^{n},
\end{align}
with equality if and only if $(x_{1},\ldots,x_{n})=s-2_{n}$. On the other hand,
\begin{align}\label{eqn:lower-all-e}
f_{\text{E},c}(x)\geq 0\qquad\forall\,x\in\mathcal{C}_{s}\times[0,1]^{n},
\end{align}
with equality if and only if $(x_{n+1},\ldots,x_{2n})=c$. Thus $f(x)=f_{\text{M}}(x)+f_{\text{E},c}(x)\geq S$ for all $x\in\mathcal{C}_{s}\times[0,1]^{n}$, with equality if and only if $x=(x_{1},\ldots,x_{n},x_{n+1},\ldots,x_{2n})=(s-2_{n},c)$.

If we can solve this optimization problem with an output $\tilde{x}$ satisfying \eqn{approx-all}, then
\begin{align}\label{eqn:approx-all-2}
f_{\text{M}}(\tilde{x})+f_{\text{E},c}(\tilde{x})=f(\tilde{x})\leq S+\tfrac{1}{3}.
\end{align}
Eqs. \eqn{lower-all-m}, \eqn{lower-all-e}, and \eqn{approx-all-2} imply
\begin{align}
f_{\text{M}}(\tilde{x})&\leq S+\tfrac{1}{3}=\min_{x\in\mathcal{C}_{s}\times[0,1]^{n}}f_{\text{M}}(x)+\tfrac{1}{3}; \label{eqn:main-all-membership} \\
f_{\text{E},c}(\tilde{x})&\leq\tfrac{1}{3}=\min_{x\in\mathcal{C}_{s}\times[0,1]^{n}}f_{\text{E},c}(x)+\tfrac{1}{3}. \label{eqn:main-all-evaluation}
\end{align}

On the one hand, Eq. \eqn{main-all-membership} says that $\tilde{x}$ also minimizes $f_{\text{M}}$ with approximation error $\epsilon=\tfrac{1}{3}$. By \thm{main-membership}, this requires $\Omega(\sqrt{n})$ queries to the membership oracle $O_{\mathcal{C}_{s}}$. Also notice that one query to $O_{\mathcal{C}_{s}\times[0,1]^{n}}$ can be trivially simulated one query to $O_{\mathcal{C}_{s}}$; therefore, minimizing $f$ with approximation error $\epsilon=\tfrac{1}{3}$ with success probability 0.9 requires $\Omega(\sqrt{n})$ quantum queries to $O_{\mathcal{C}_{s}\times[0,1]^{n}}$.

On the other hand, Eq. \eqn{main-all-evaluation} says that $\tilde{x}$ minimizes $f_{\text{E},c}$ with approximation error $\epsilon=\tfrac{1}{3}$. By \thm{main-evaluation}, it takes $\Omega(\sqrt{n}/\log n)$ queries to $O_{f_{\text{E},c}}$ to output $\tilde{x}$. Also notice that
\begin{align}
f(x)=f_{\text{M}}(x)+f_{\text{E},c}(x)=\sum_{i=1}^{n}x_{i}+f_{\text{E},c}(x);
\end{align}
therefore, one query to $O_{f}$ can be simulated by one query to $O_{f_{\text{E},c}}$. Therefore, approximately minimizing $f$ with success probability 0.9 requires $\Omega(\sqrt{n}/\log n)$ quantum queries to $O_{f}$.

In addition, $f_{\text{M}}$ is independent of the coordinates $x_{n+1},\ldots,x_{2n}$ and only depends on the coordinates $x_{1},\ldots,x_{n}$, whereas $f_{\text{E},c}$ is independent of the coordinates $x_{1},\ldots,x_{n}$ and only depends on the coordinates $x_{n+1},\ldots,x_{2n}$. As a result,  the oracle $O_{\mathcal{C}_{s}\times[0,1]^{n}}$ reveals no information about $c$, and $O_f$ reveals no information about $s$.
Since solving the optimization problem reveals both $s$ and $c$, the lower bounds on query complexity must hold simultaneously.

Overall, to output an $\tilde{x}\in\mathcal{C}_{s}\times[0,1]^{n}$ satisfying \eqn{approx-all} with success probability at least $0.9\cdot 0.9>0.8$, we need $\Omega(\sqrt{n})$ quantum queries to $O_{\mathcal{C}_{s}\times[0,1]^{n}}$ and $\Omega(\sqrt{n}/\log n)$ quantum queries to $O_{f}$, as claimed.
\end{proof}

\subsection{Smoothed hypercube}\label{sec:smooth}

As a side point, our quantum lower bound in \thm{main-all} also holds for a smooth convex body. Given an $n$-dimensional hypercube $\mathcal{C}_{x,l}:=\bigtimes_{i=1}^{n}[x_{i}-l,x_{i}]$, we define a smoothed version as
\begin{align}\label{eqn:smoothed}
\mathcal{SC}_{x,l}:=B_{2}\bigg(\bigtimes_{i=1}^{n}\Big[x_{i}-\frac{2n}{2n+1}l,x_{i}-\frac{1}{2n+1}l\Big],\frac{1}{2n+1}l\bigg)
\end{align}
using \defn{conv_nbd}. For instance, a smoothed 3-dimensional cube is shown in \fig{smoothed}.

\begin{figure}[htbp]
\centering
\includegraphics[width=2.5in]{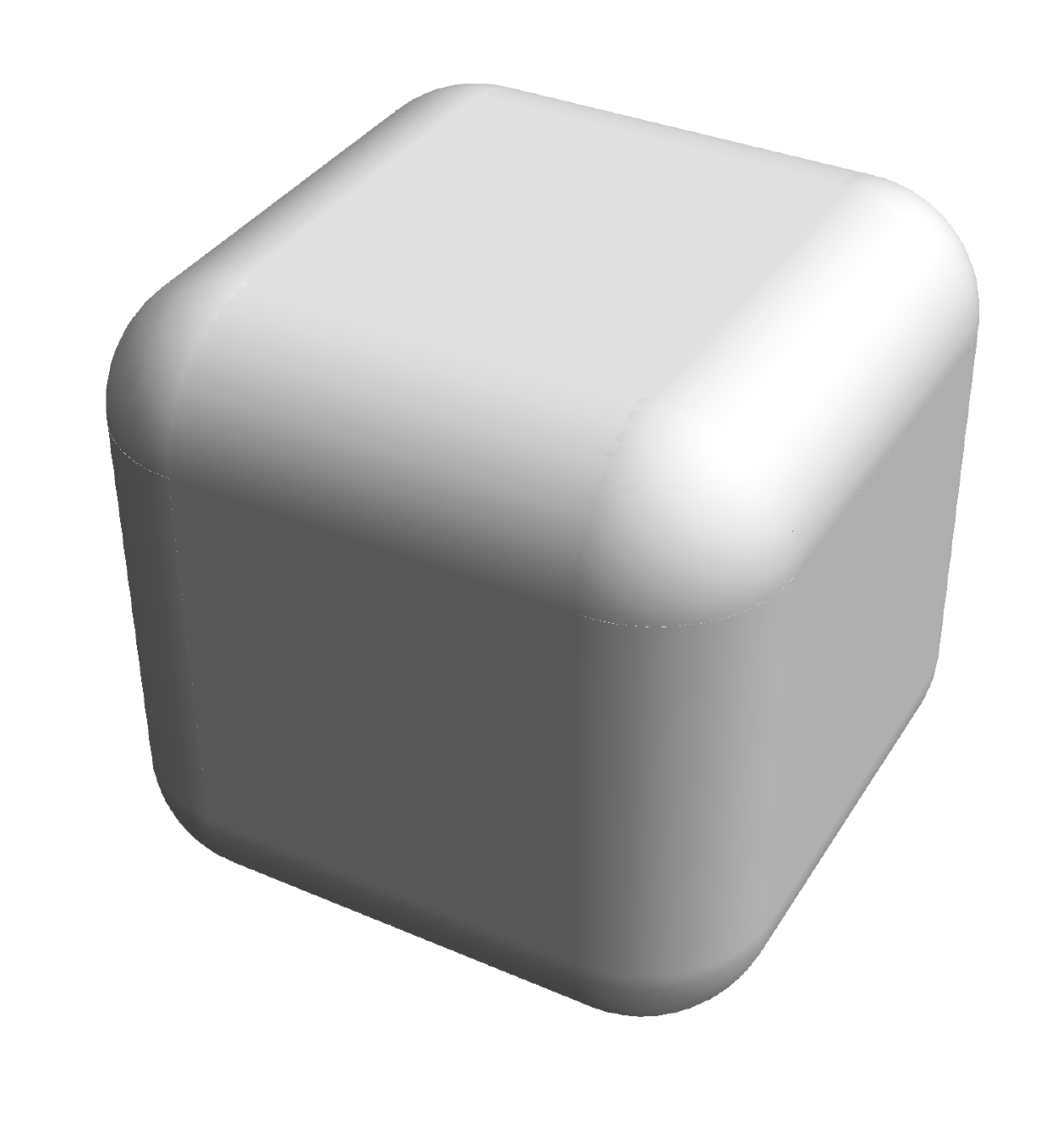}
\caption{Smoothed hypercube of dimension 3.}
\label{fig:smoothed}
\end{figure}

The smoothed hypercube satisfies
\begin{align}
\mathcal{C}_{x-\frac{1}{2n+1}l_{n},\frac{2n-1}{2n+1}l}\subseteq \mathcal{SC}_{x,l}\subseteq \mathcal{C}_{x,l}
\end{align}
where $l_{n}$ is $l$ times the $n$-dimensional all-ones vector; in other words, it is contained in the original (non-smoothed) hypercube, and it contains the hypercube with the same center but edge length $\frac{2n-1}{2n+1}l$. For instance, $\bigtimes_{i=1}^{n}[\frac{1}{2n+1},\frac{2n}{2n+1}]\subseteq \mathcal{SC}_{1_{n},1}\subseteq \bigtimes_{i=1}^{n}[0,1]$; by Eq. \eqn{discrete-set}, $D_{n}\subseteq\mathcal{SC}_{1_{n},1}$. It can be verified that the proof of \thm{main-membership} still holds if the hypercube $\bigtimes_{i=1}^{n}[s_{i}-2,s_{i}+1]=\mathcal{C}_{s+1_{n},3}$ is replaced by $\mathcal{SC}_{s+1_{n},3}$, and the proof of \thm{main-evaluation} still holds if the unit hypercube $[0,1]^{n}$ is replaced by $\mathcal{SC}_{1_{n},1}$; consequently \thm{main-all} also holds. More generally, the proofs remain valid as long as the smoothed hypercube is contained in $[0,1]^{n}$ and contains $D_{n}$ (for discretization).


\section*{Acknowledgements}
We thank Yin Tat Lee for numerous helpful discussions and anonymous reviewers for suggestions on preliminary versions of this paper. We also thank Joran van Apeldoorn, Andr{\'a}s Gily{\'e}n, Sander Gribling, and Ronald de Wolf for sharing a preliminary version of their manuscript \cite{vanApeldoorn2018optimization} and for their detailed feedback on a preliminary version of this paper, including identifying some mistakes in previous lower bound arguments and pointing out a minor technical issue in the evaluation of the height function in \lin{mem-to-sep-6} of \algo{halfspace} (and in \cite{lee2017efficient}). This work was supported in part by the U.S. Department of Energy, Office of Science, Office of Advanced Scientific Computing Research, Quantum Algorithms Teams program. AMC also received support from the Army Research Office (MURI award W911NF-16-1-0349), the Canadian Institute for Advanced Research, and the National Science Foundation (grant CCF-1526380). XW also received support from the National Science Foundation (grants CCF-1755800 and CCF-1816695).


\providecommand{\bysame}{\leavevmode\hbox to3em{\hrulefill}\thinspace}


\begin{appendix}

\section{Auxiliary lemmas}\label{sec:auxiliary}
\subsection{Classical gradient computation}\label{sec:classical-gradient}

Here we prove that the classical query complexity of gradient computation is linear in the dimension.

\begin{lemma}
  Let $f$ be an $L$-Lipschitz convex function that is specified by an evaluation oracle with precision $\delta = 1/\poly(n)$. Any (deterministic or randomized) classical algorithm to calculate a subgradient of $f$ with $L_\infty$-norm error $\epsilon = 1/\poly(n)$ must make $\tilde{\Omega}(n)$ queries to the evaluation oracle.
\end{lemma}

\begin{proof}
  Consider the linear function $f(x) = c^Tx$ where each $c_i \in [0,1]$. Since each $c_i$ must be determined to precision $\epsilon$, the problem hides $n \log (1/\epsilon)$ bits of information. Furthermore, since the evaluation oracle has precision $\delta$, each query reveals only $\log (1/\delta)$ bits of information. Thus any classical algorithm must make at least $\frac{n \log(1/\epsilon)}{\log(1/\delta)} = n/\log(n)$ evaluation queries.
\end{proof}

\subsection{Mollified functions}
The following lemma establishes properties of mollified functions:

\begin{lemma}[Mollifier properties] \label{lem:mollify}
  Let $f\colon \R^n \to \R$ be an $L$-Lipschitz convex function with mollification $F_\delta = f \ast m_\delta$, where $m_\delta$ is defined in \eq{mollifier}. Then
  \begin{enumerate}[nosep,label=\normalfont(\roman*)]
  \item $F_\delta$ is infinitely differentiable,
  \item $F_\delta$ is convex,
  \item $F_\delta$ is $L$-Lipschitz continuous, and
  \item $|F_\delta(x) - f(x)| \le L\delta$.
  \end{enumerate}
\end{lemma}
\begin{proof}~
  \begin{enumerate}[label=\normalfont(\roman*)]
  \item Convolution satisfies $\frac{\d(p \ast q)}{\d{x}} = p \ast \frac{\d{q}}{\d{x}}$, so because $m_\delta$ is infinitely differentiable, $F_\delta$ is infinitely differentiable.
  \item We have $F_\delta(x) = \int_{\R^n} f(x - z)m_\delta(z) \,\d{z} = \int_{\R^n} f(z)m_\delta(x - z) \,\d{z}
    $. Thus
    \begin{align}
      F_\delta(\lambda x + (1 - \lambda) y) &= \int\limits_{\R^n} f(\lambda x + (1 - \lambda) y - z)m_\delta(z) \,\d{z} \\
                                            &\ge \int\limits_{\R^n} [\lambda f(x - z) + (1 - \lambda) f(y - z)]m_\delta(z) \,\d{z} \\
                                            & = \lambda F_\delta(x) + (1 - \lambda) F_\delta(y),
    \end{align}
    where the inequality holds by convexity of $f$ and the fact that $m_\delta \ge 0$. Thus $F_\delta$ is convex.
  \item We have
    \begin{align}
      \norm{ F_\delta(x) - F_\delta(y) } &= \norm{\, \int\limits_{\R^n} [f(x - z) - f(y-z)] m_\delta(z) \,\d{z} } \\
                                         &\le \int\limits_{\R^n} \norm{f(x - z) - f(y-z)} m_\delta(z) \,\d{z} \\
                                         &\le L \norm{x - y} \int\limits_{\R^n} m_\delta(z) \,\d{z} \\
                                         & = L \norm{x - y}.
    \end{align}
    Thus from \defn{lipschitz}, $F_\delta$ is $L$-Lipschitz.
  \item We have
    \begin{align}
      \left| F_\delta(x) - f(x) \right| &= \left| \, \int\limits_{\R^n} f(x - z) g(z) \,\d{z} - \int\limits_{\R^n} f(x) g(z) \,\d{z} \right| \\
                                        &\le \int\limits_{\R^n} \left| f(x - z) - f(z) \right| g(z) \,\d{z} \\
                                        &\le L \int\limits_{\R^n} \left| z \right| g(z) \,\d{z} \\
                                        &= L\int\limits_{B_2(0,\delta)} \frac{|z|}{I_n}\exp\left(-\frac{1}{1 - \norm{z/\delta}^2}\right) \,\d{z} \\
                                        &= L\delta\int\limits_{B_2(0,1)} \frac{|u|}{I_n}\exp\left(-\frac{1}{1 - \norm{u}^2}\right) \,\d{u} \\
      & \le L\delta\int\limits_{B_2(0,1)} \frac{1}{I_n}\exp\left(-\frac{1}{1 - \norm{u}^2}\right) \,\d{u} \\
                                        &= L\delta
    \end{align}
    as claimed.
  \end{enumerate}
\end{proof}

The following lemma shows strong convexity of mollified functions, ruling out the possibility of directly applying \lem{jordan_whp_detailed} to calculate subgradients.

\begin{lemma}
  \label{lem:non-smooth-differentiable}
  There exists a $1$-Lipschitz convex function $f$ such that for any $\beta$-smooth function $g$ with $|f(x) - g(x)| \le \delta$ for all $x$,
  $\beta\delta \ge c$ where $c$ is a constant.
\end{lemma}

\begin{proof}
Let $f(x) = |x|$. Consider $x \ge 0$. By the smoothness of $g$,
\begin{align}
  \label{eq:1Lipschitz-1}
  g(x) &\le g(0) + \nabla g(0)^Tx + \frac{\beta}{2} x^2, \\
  g(-x) &\le g(0) - \nabla g(0)^Tx + \frac{\beta}{2} x^2.
\end{align}
As a result, we have $g(x) + g(-x) \le 2g(0) + \beta x^2$ for all $x > 0$. Since $|f(x) - g(x)| \le \delta$,
\begin{align}
  \label{eq:1Lipschitz-3}
  f(x) + f(-x) &\le 2f(0) + \beta x^2 + 4\delta \quad  \hence\quad \beta x^2  - 2x + 4\delta \ge 0
\end{align}
for all $x > 0$.

Since $4\delta > 0$, the discriminant must be non-positive. Therefore, $16 - 16\beta\delta \le 0$, so $\beta\delta\geq 1$.
\end{proof}


\section{Proof details for upper bound}\label{sec:upper-bound-appendix}
We give the complete proof of \lem{jordan_whp} in this section.

Given a quantum oracle that computes the function $N_0F$ in the form
\begin{align}
U_F\ket{x}\ket{y} = \ket{x}\ket{y \oplus (N_0F(x) \mod N)},
\end{align}
it is well known that querying $U_F$ with
\begin{align}
\ket{y_0} = \frac{1}{\sqrt{N_0}}\sum_{i \in \{0,1,\ldots,N-1\}} e^{\frac{2\pi i x}{N_0}}\ket{i}
\end{align}
allows us to implement the phase oracle $O_F$ in one query. This is a common technique used in quantum algorithms known as \emph{phase kickback}.

First, we prove the following lemma:
\begin{lemma}
  \label{lem:fixed}
  Let $G:=\lbrace -N/2,-N/2+1,\dots,N/2-1 \rbrace$ and define $\gamma\colon \lbrace 0,1,\dots,N-1 \rbrace \to G$ by $\gamma(x) = x - N/2$ for all $x \in \lbrace 0,1,\dots,N-1 \rbrace$. Consider the inverse quantum Fourier transforms
  \begin{align}
    \QFT_N^{-1}\ket{x}&:= \frac{1}{\sqrt{N}}\sum_{y \in [0,N-1]}e^{-\frac{2\pi i xy}{N}}\ket{y}, \quad \forall\,x \in [0,N-1]; \\
    \QFT_G^{-1}\ket{\gamma(x)}&:= \frac{1}{\sqrt{N}}\sum_{\gamma(y) \in G}e^{-\frac{2\pi i \gamma(x)\gamma(y)}{N}}\ket{\gamma(y)}, \quad \forall\,\gamma(x) \in G
  \end{align}
  over $[0,N-1]:= \{0,1,\ldots,N-1\}$ and $G$, respectively.
  Then we have
  $\QFT_G^{-1} = U \QFT_N^{-1} U$, where $U$ is a tensor product of $b = \log_2{N}$ single-qubit unitaries.
\end{lemma}

\begin{proof}
  For any $x \in [0,N-1]$, we have
  \begin{align}
    \QFT_G^{-1}\ket{x} = \frac{1}{\sqrt{N}}\sum_{y \in [0,N-1]}e^{-\frac{2\pi i \gamma(x)\gamma(y)}{N}}\ket{y}
  \end{align}
  which is equivalent to
  \begin{align}
    \frac{1}{\sqrt{N}}\sum_{y \in [0,N-1]}e^{-\frac{2\pi i xy}{N}}e^{\pi i (x + y)}\ket{y}
  \end{align}
  up to a global phase.
  Setting $U\ket{x} = e^{\pi i x} \ket{x}$ for all $x \in \{0,1,\ldots,N-1\}$, we have the result.
\end{proof}
The above shows that we can implement $\QFT_G^{-1}$ on a single $b$-bit register using $O(b)$ gates. Thus there is no significant overhead in gate complexity that results from using $\QFT_G$ instead of the usual QFT.

Now we prove \lem{jordan_whp}, which is rewritten below:
\begin{lemma}
  \label{lem:jordan_whp_detailed}
  Let $f\colon \R^n \to \R$ be an $L$-Lipschitz function that is specified by an evaluation oracle with error at most $\epsilon$. Let $f$ be $\beta$-smooth in $B_\infty(x,2\sqrt{{\epsilon}/{\beta}})$, and let $\tilde{g}$ be the output of $\algname{GradientEstimate}(f,\epsilon,L,\beta,x_0)$ (from \algo{grad_est}). Let $g = \nabla f(x_0)$. Then
  \begin{align}
    \Pr\left[ |\tilde{g}_i - g_i| > 1500\sqrt{n\epsilon\beta} \right] < \frac{1}{3}, \quad \forall\,i \in \range{n}.
  \end{align}
\end{lemma}
\begin{proof}
  To analyze the \algname{GradientEstimate} algorithm, let the actual state obtained before applying the inverse QFT over $G$ be

\begin{equation}
  \label{eq:real}
  \ket{\psi} = \frac{1}{N^{n/2}}\sum_{x \in G^d} e^{2 \pi i \tilde{F}(x)} \ket{x},
\end{equation}
where $|\tilde{F}(x)-\frac{N}{2Ll}[f(x_0 + \frac{lx}{N}) - f(x_0)]| \le \frac{1}{N_0}$.  Also consider the idealized state
\begin{align}
  \label{eq:ideal}
  \ket{\phi} &= \frac{1}{N^{n/2}}\sum_{x \in G^d} e^{\frac{2 \pi i g \cdot x}{2L}}\ket{x}.
\end{align}

From \lem{fixed} we can efficiently apply the inverse QFT over $G$; from the analysis of phase estimation (see \cite{brassard2002quantum}), we know that
\begin{align}
  \label{eq:phase_est}
  \forall\,i \in \range{n} \quad \Pr \left[ \Big|\frac{N g_i}{2L} - k_i\Big| > w \right] < \frac{1}{2(w-1)},
\end{align}
so in particular,
\begin{align}
  \forall\,i \in \range{n} \quad \Pr \left[ \Big|\frac{N g_i}{2L} - k_i\Big| > 4 \right] < \frac{1}{6}.
\end{align}

Now, let $g = \nabla{f}(x_0)$. The difference in the probabilities of any measurement on $\ket{\psi}$ and $\ket{\phi}$ is bounded by the trace distance between the two density matrices, which is
\begin{equation}
  \label{eq:tr_distance}
  \norm {\ketbra{\psi}{\psi} - \ketbra{\phi}{\phi}}_1 = 2\sqrt{1 - |\braket{\psi}{\phi}|^2} \le 2\norm{\ket\psi - \ket\phi}.
\end{equation}
Since $f$ is $\beta$-smooth in $B_\infty(x,2\sqrt{\frac{\epsilon}{\beta}})$,

\begin{align}
  \label{eq:expansion}
  \tilde{F}(x) &\le \frac{N}{2Ll}\left[f\left(x_0 + \frac{lx}{N}\right) - f(x_0)\right] + \frac{1}{N_0} \\
               &\le \frac{N}{2Ll}\left( \frac{l}{N}\nabla f(x_0) \cdot x + \frac{\beta l^2 x^2}{2 N^2} \right) + \frac{1}{N_0}\\
               &\le \frac{1}{2L}\nabla f(x_0) \cdot x + \frac{N\beta ln}{4L} + \frac{N\epsilon}{Ll}.
\end{align}
Then we have
\begin{align}
  \label{eq:bound}
  \norm{\ket\psi - \ket\phi}^2 &= \frac{1}{N^d} \sum_{x \in G_d} |e^{2\pi i\tilde{F}(x)} - e^ {\frac{2\pi i g.x}{2L}}|^2\\
                               &= \frac{1}{N^d} \sum_{x \in G_d} 4\pi^2\left(\tilde{F}(x) - \frac{g \cdot x}{2L}\right)^2 \\
                               &\le \frac{1}{N^d} \sum_{x \in G_d} \frac{4 \pi^2 N^2}{L^2}\left(\frac{\beta ln}{4} + \frac{\epsilon}{l}\right)^2\\
                               &= \frac{4 \pi^2 N^2 \beta \epsilon n}{L^2}.
\end{align}
In \algo{grad_est}, $N$ is chosen such that $N \le \frac{L}{24\pi\sqrt{n\epsilon\beta}}$. Plugging this into \eq{bound},
\begin{equation}
  \label{eq:trace-distance-bound}
  \norm{\ket\psi - \ket\phi}^2 \le \frac{1}{144}.
\end{equation}
Thus the trace distance is at most $\frac{1}{6}$. Therefore, $\Pr\left[\big|k_i - \frac{Ng_i}{2L}\big| > 4\right] < \frac{1}{3}$. Thus we have
  \begin{equation}
    \Pr\left[ |\tilde{g}_i - g_i| > \frac{8L}{N} \right] < \frac{1}{3}, \quad \forall i \in \range{n}.
  \end{equation}
  From \algo{grad_est}, $\frac{1}{N} \le \frac{48\pi\sqrt{n\epsilon\beta}}{L}$, so $\frac{8L}{N} < 384\pi\sqrt{n\epsilon\beta} < 1500\sqrt{n\epsilon\beta}$, and the result follows.
\end{proof}

Finally, we prove that the height function $h_p$ can be evaluated with precision $\epsilon$ using $O(\log(1/\epsilon))$ queries to a membership oracle:
\begin{lemma}
  \label{lem:height-eval}
  The function $h_p(x)$ can be evaluated for any $x \in B_\infty(0,r/2)$ with any precision $\epsilon \ge 7\kappa\delta$ using $O(\log(1/\epsilon))$ queries to a membership oracle with error $\delta$.
\end{lemma}

\begin{figure}[htbp]
\centering
\begin{tikzpicture}[scale=1.8]
\coordinate (mid) at (2,-0.25);
\coordinate (up) at (2.5,0.0625);
\coordinate (down) at (2.5,-0.25);
\node at (2.2,-0.19) {\tiny $\theta$};
\node at (2.15,-0.4) {\small $Q$};

\node at (2.15,-1.2) {\small $H$};
\draw [line width=0.1mm, black ] (2.03,-0.95) -- ((2.1,-1.1) {};;

\coordinate (mid2) at (-0.75,-1.5);
\coordinate (up2) at (-0.25,-1.1875);
\coordinate (down2) at (-0.25,-1.5);
\node at (-0.55,-1.44) {\tiny $\theta$};

\node at (-0.4,-1.64) {\footnotesize $-\Delta\vec{q}$};

\node at (-0.9,-1.6) {\small $y$};

\draw[thick,black,->] (0,-1.5) -- (4,1) node[black,right] {$\vec{p}$};

    \draw [dashed] (1.5,-2) -- (1.5,1) {};;
    \draw [line width=0.25mm, black ] (2,-2) -- (2,1) {};;
    \draw [dashed] (2.5,-2) -- (2.5,1) {};;
    \draw [line width=0.25mm, black ] (0,-1.5) -- (-0.75,-1.5) {};;
    \draw [fill] (0,-1.5) circle [radius=0.05] node [below right] {\small $x$};;
    \draw [line width=0.25mm, black ] (-0.75,-1.5) -- (2.5,0.53125) {};;
    \draw [line width=0.1mm, red ] (2.4,-0.05) -- (2.6,-0.1) {};;
    \draw [line width=0.25mm, red ] (1.5,-0.5625) -- (2.5,0.0625) node [below right] {\footnotesize error};;
    \draw [line width=0.25mm, dashed, blue,-> ] (2,-0.25) -- (2.5,-0.25) node [below right] {\footnotesize $\vec{q}$};;
    \draw [line width=0.1mm, blue ] (2.3,-0.3) -- (2.55,-0.4) {};;
    \draw [line width=0.25mm, black ] (0,-1.5) -- (0,-1.03125) {};;
    \draw [densely dotted, black,<->] (1.5,0.75) -- (2,0.75) ;;
    \draw [densely dotted, black,<->] (2,0.75) -- (2.5,0.75) ;;
    \node at (1.75,0.85) {\tiny$\delta$};
    \node at (2.25,0.85) {\tiny$\delta$};
    \node at (-0.45,-1.025) {\footnotesize$\frac{\Delta}{\cos\theta}$};

    \pic [draw, -, angle eccentricity=1.5, angle radius=0.2cm] {angle = down--mid--up};;
    \pic [draw, -, angle eccentricity=1.5, angle radius=0.2cm] {angle = down2--mid2--up2};
\end{tikzpicture}
\caption{Relating the error to $2\delta$ in $n=2$ dimensions.}
\label{fig:height-eval-figure}
\end{figure}
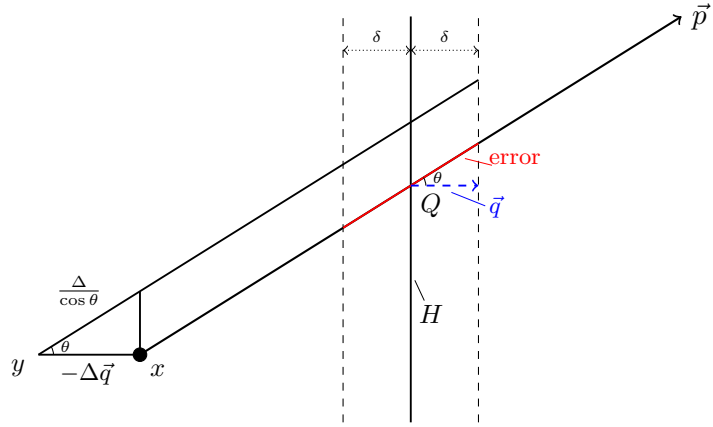

\begin{proof}
We denote the intersection of the ray $x + t \vec{p}$ and the boundary of $K$ by $Q$, and let $H$ be an $(n-1)$-dimensional hyperplane that is tangent to $K$ at $Q$. Because $K$ is convex, it lies on only one side of $H$; we let $\vec{q}$ denote the unit vector at $Q$ that is perpendicular to $H$ and points out of $K$. Let $\theta:=\arccos\<\vec{p},\vec{q}\>$.

Using binary search with $\log(1/\delta)$ queries, we can find a point $P$ on the ray $x + t \vec{p}$ such that $P \notin B(K,-\delta)$ and $P \in B(K,\delta)$. The total error for $t$ is then at most $\frac{2\delta}{\cos{\theta}}$. Now consider $y = x-\Delta\vec{q}$ for some small $\Delta>0$. Then $h_{p}(y)-h_{p}(x)=\frac{\Delta}{\cos{\theta}}+o(\frac{\Delta}{\cos{\theta}})$ (see \fig{height-eval-figure} for an illustration with $n=2$).

By \prop{height-function-lipschitz}, $h_p(x)$ is $3\kappa$-Lipschitz for any $x \in B(0,r/2)$; therefore, $h_{p}(y)-h_{p}(x)\leq 3\kappa\|y-x\|=3\kappa\Delta$, and hence
\begin{align}
\frac{\Delta}{\cos{\theta}}+o\Big(\frac{\Delta}{\cos{\theta}}\Big)\leq 3\kappa\Delta\quad\Rightarrow\quad \frac{1}{\cos\theta}\leq 3.5\kappa
\end{align}
for a small enough $\Delta>0$. Thus the error in $h_p(x)$ is at most $\frac{2\delta}{\cos{\theta}}\leq 7\kappa\delta$, and the result follows.
\end{proof}

\section{Proof details for lower bound} \label{sec:lower-bound-appendix}
In this section, we give proof details for our claims in \sec{lower-evaluation}.

\subsection{Convexity of max-norm optimization}\label{sec:evaluation-convexity}
In this subsection, we prove:
\begin{lemma}\label{lem:evaluation-convexity}
The function
\begin{align}
f_{c}(x)=\max_{i\in\range{n}}|\pi(x_{i})-c_{i}|+\Big(\sum_{i=1}^{n}|\pi(x_{i})-x_{i}|\Big)
\end{align}
is convex on $\R^{n}$, where $c\in\{0,1\}^{n}$ and $\pi\colon\R\rightarrow[0,1]$ is defined as
\begin{align}\label{eqn:projection-definition-appendix}
\pi(x)=\begin{cases}
    0 & \text{if $x<0$} \\
    x & \text{if $0\leq x\leq 1$} \\
    1 & \text{if $x>1$}.
  \end{cases}
\end{align}
\end{lemma}

\begin{proof}
For convenience, we define $g_i \colon\R^n \to\R$ for $i \in [n]$ as
\begin{align}\label{eqn:function-g-definition}
g_i(x):=|\pi(x_i)-x_i|=\begin{cases}
    -x_i & \text{if $x_i<0$} \\
    0 & \text{if $0\leq x_i\leq 1$} \\
    x_i-1 & \text{if $x_i>1$}
  \end{cases}
\end{align}
where the second equality follows from \eqn{projection-definition-appendix}. It is clear that $g_i(x) = \max \{-x_i, 0, x_i-1\}$ by \eqn{function-g-definition}. Since the pointwise maximum of convex functions is convex, $g_i(x)$ is convex for all $i \in [n]$.

Moreover, for all $i\in\range{n}$ we define $h_{c,i}\colon\R^n\to\R$ as $h_{c,i}(x):=|\pi(x_{i})-c_{i}|+|\pi(x_{i})-x_{i}|$. We claim that $h_{c,i}(x)=|x_{i}-c_{i}|$, and thus $h_{c,i}$ is convex. If $c_{i}=0$, then $|\pi(x_{i})-c_{i}|+|\pi(x_{i})-x_{i}|=\pi(x_{i})+|\pi(x_{i})-x_{i}|$; as a result,
\begin{align}
x_{i}<0 &\quad\Rightarrow\quad \pi(x_{i})+|\pi(x_{i})-x_{i}|=0+|0-x_{i}|=-x_{i}; \\
0\leq x_{i}\leq 1&\quad\Rightarrow\quad \pi(x_{i})+|\pi(x_{i})-x_{i}|=x_{i}+|x_{i}-x_{i}|=x_{i}; \\
x_{i}>1 &\quad\Rightarrow\quad \pi(x_{i})+|\pi(x_{i})-x_{i}|=1+|1-x_{i}|=x_{i}.
\end{align}
Therefore, $\forall\,i \in\range{n}, h_{c,i}(x)=|x_{i}-c_{i}|$. The proof is similar when $c_{i}=1$.

Now we have
\begin{align}
f_{c}(x)&=\max_{i\in\range{n}}\Big(|\pi(x_{i})-c_{i}|+\sum_{j=1}^{n}|\pi(x_{j})-x_{j}|\Big) \\
&=\max_{i\in\range{n}}\Big(\big(|\pi(x_{i})-c_{i}|+|\pi(x_{i})-x_{i}|\big)+\sum_{j\neq i}g_j(x)\Big) \\
&=\max_{i\in\range{n}}\Big(h_{c,i}(x)+\sum_{j\neq i}g_j(x)\Big).
\end{align}
Because $h_{c,i}$ and $g_j$ are both convex functions on $\R^n$ for all $i,j\in\range{n}$, the function $h_{c,i}(x)+\sum_{j\neq i}g_j(x)$ is convex on $\R^{n}$. Thus $f_c$ is the pointwise maximum of $n$ convex functions and is therefore itself convex.
\end{proof}

\subsection{Proof of \lem{discretization}}\label{sec:discretization-appendix}
\subsubsection{Correctness of \lin{discretization-1} and \lin{discretization-2}}\label{sec:discrete-line-12}

In this subsection, we prove:

\begin{lemma}\label{lem:discrete-line-12}
Let $b$ and $\sigma$ be the values computed in \lin{discretization-1} of \algo{discretization}, and let $x^{*} = \chi(b,\sigma^{-1})$. Then $\Ord(x^{*})=\Ord(x)$.
\end{lemma}

\begin{proof}
First, observe that $b\in\{0,1\}^{n}$ and $\sigma\in S_{n}$ because
\begin{itemize}
\item For all $i\in\range{n}$, both $x_{i}$ and $1-x_{i}$ can be written as $b_{i}x_{i}+(1-b_{i})(1-x_{i})$ for some $b_{i}\in\{0,1\}$;
\item $\Ord(x)$ is \emph{palindrome}, i.e., if $x_{i_{1}}$ is the largest in $\{x_{1},\ldots,x_{n},1-x_{1},\ldots,1-x_{n}\}$ then $1-x_{i_{1}}$ is the smallest in $\{x_{1},\ldots,x_{n},1-x_{1},\ldots,1-x_{n}\}$; if $1-x_{i_{2}}$ is the second largest in $\{x_{1},\ldots,x_{n},1-x_{1},\ldots,1-x_{n}\}$ then $x_{i_{2}}$ is the second smallest in $\{x_{1},\ldots,x_{n},1-x_{1},\ldots,1-x_{n}\}$; etc.
\end{itemize}

Recall that in \eqn{ord-x}, the decreasing order of $\{x_{1},\ldots,x_{n},1-x_{1},\ldots,1-x_{n}\}$ is
\begin{align}\label{eqn:ord-x-rewrite}
&b_{\sigma(1)}x_{\sigma(1)}+(1-b_{\sigma(1)})(1-x_{\sigma(1)})\geq\cdots\geq b_{\sigma(n)}x_{\sigma(n)}+(1-b_{\sigma(n)})(1-x_{\sigma(n)}) \nonumber \\
&\quad \geq (1-b_{\sigma(n)})x_{\sigma(n)}+b_{\sigma(n)}(1-x_{\sigma(n)})\geq\cdots\geq (1-b_{\sigma(1)})x_{\sigma(1)}+b_{\sigma(1)}(1-x_{\sigma(1)}).
\end{align}
On the other hand, by the definition of $D_{n}$, we have
\begin{align}\label{eqn:ord-x*}
\{x^{*}_{1},\ldots,x^{*}_{n},1-x^{*}_{1},\ldots,1-x^{*}_{n}\}=\Big\{\frac{1}{2n+1},\frac{2}{2n+1},\ldots,\frac{2n}{2n+1}\Big\}.
\end{align}
Combining \eqn{ord-x-rewrite} and \eqn{ord-x*}, it suffices to prove that for any $i\in\range{n}$,
\begin{align}
b_{\sigma(i)}x^{*}_{\sigma(i)}+(1-b_{\sigma(i)})(1-x^{*}_{\sigma(i)})&= 1-\frac{i}{2n+1}; \label{eqn:ord-x-1} \\
(1-b_{\sigma(i)})x^{*}_{\sigma(i)}+b_{\sigma(i)}(1-x^{*}_{\sigma(i)})&= \frac{i}{2n+1}. \label{eqn:ord-x-2}
\end{align}
We only prove \eqn{ord-x-1}; the proof of \eqn{ord-x-2} follows symmetrically.

By \eq{chi}, we have $x_{j}^{*}=(1-b_{j})\frac{\sigma^{-1}(j)}{2n+1}+b_{j}(1-\frac{\sigma^{-1}(j)}{2n+1})$ for all $j\in\range{n}$; taking $j=\sigma(i)$, we have $x_{\sigma(i)}^{*}=(1-b_{\sigma(i)})\frac{i}{2n+1}+b_{\sigma(i)}(1-\frac{i}{2n+1})$. Moreover, since $b_{\sigma(i)}\in\{0,1\}$ implies that $b_{\sigma(i)}(1-b_{\sigma(i)})=0$ and $b_{\sigma(i)}^{2}+(1-b_{\sigma(i)})^{2}=1$, we have
\begin{align}
&b_{\sigma(i)}x^{*}_{\sigma(i)}+(1-b_{\sigma(i)})(1-x^{*}_{\sigma(i)}) \nonumber \\
&\qquad=b_{\sigma(i)}\big[(1-b_{\sigma(i)})\tfrac{i}{2n+1}+b_{\sigma(i)}\big(1-\tfrac{i}{2n+1}\big)\big] \nonumber \\
&\qquad\quad+(1-b_{\sigma(i)})\big[b_{\sigma(i)}\tfrac{i}{2n+1}+(1-b_{\sigma(i)})\big(1-\tfrac{i}{2n+1}\big)\big] \\
&\qquad=2b_{\sigma(i)}(1-b_{\sigma(i)})\tfrac{i}{2n+1}+\big(b_{\sigma(i)}^{2}+(1-b_{\sigma(i)})^{2}\big)\big(1-\tfrac{i}{2n+1}\big) \\
&\qquad=1-\tfrac{i}{2n+1},
\end{align}
which is exactly \eqn{ord-x-1}.
\end{proof}

\subsubsection{Correctness of \lin{discretization-3}}\label{sec:discrete-line-3}
In this subsection, we prove:
\begin{lemma}\label{lem:discrete-line-3}
There is some $k^{*}\in\{1,\ldots,n+1\}$ such that $f(x^{*})=1-\frac{k^{*}}{2n+1}$.
\end{lemma}

\begin{proof}
Because $|x_{i}^{*}-c_{i}|$ is an integer multiple of $\frac{1}{2n+1}$ for all $i\in\range{n}$, $f(x^{*})$ must also be an integer multiple of $\frac{1}{2n+1}$. As a result, $k^{*}=(2n+1)(1-f(x^{*}))\in\Z$.

It remains to prove that $1\leq k^{*}\leq n+1$. By the definition of $D_{n}$ in \eqn{discrete-set}, we have
\begin{align}\label{eqn:discrete-line-3-1}
x^{*}_{i}=(1-b_{i})\frac{\sigma^{-1}(i)}{2n+1}+b_{i}\Big(1-\frac{\sigma^{-1}(i)}{2n+1}\Big)\qquad\forall\,i\in\range{n};
\end{align}
since $b_{i}=0$ or 1, we have $x^{*}_{i}\in\{\frac{\sigma^{-1}(i)}{2n+1},1-\frac{\sigma^{-1}(i)}{2n+1}\}$. Because we also have $c_{i}\in\{0,1\}$,
\begin{align}
|x^{*}_{i}-c_{i}|\leq 1-\frac{\sigma^{-1}(i)}{2n+1}\leq\frac{2n}{2n+1}.
\end{align}
As a result,
\begin{align}
f(x^{*})=\max_{i\in\range{n}}|x^{*}_{i}-c_{i}|\leq\frac{2n}{2n+1}\ \Rightarrow\ k^{*}\geq 1.
\end{align}

It remains to prove $k^{*}\leq n+1$. By \eqn{discrete-line-3-1}, we have
\begin{align}\label{eqn:discrete-line-3-2}
x^{*}_{\sigma(n)}\in\Big\{\frac{n}{2n+1},\frac{n+1}{2n+1}\Big\};
\end{align}
because $c_{\sigma(n)}\in\{0,1\}$, we have
\begin{align}\label{eqn:discrete-line-3-3}
|x^{*}_{\sigma(n)}-c_{\sigma(n)}|\geq\frac{n}{2n+1}.
\end{align}
Therefore, we have
\begin{align}\label{eqn:discrete-line-3-4}
f(x^{*})=\max_{i\in\range{n}}|x^{*}_{i}-c_{i}|\geq|x^{*}_{\sigma(n)}-c_{\sigma(n)}|\geq\frac{n}{2n+1},
\end{align}
which implies $k^{*}\leq n+1$.
\end{proof}

\subsubsection{Correctness of \lin{discretization-4}}\label{sec:discrete-line-4}
In this subsection, we prove:
\begin{lemma}\label{lem:discrete-line-4}
The output of $f(x)$ in \lin{discretization-4} is correct.
\end{lemma}

\begin{proof}
A key observation we use in the proof, following directly from \eqn{discrete-line-3-1}, is that
\begin{align}\label{eqn:discrete-line-4-1}
|x^{*}_{\sigma(i)}-c_{\sigma(i)}|=
\begin{cases}
\frac{i}{2n+1} & \text{ if }c_{\sigma(i)}=b_{\sigma(i)}; \\
1-\frac{i}{2n+1} & \text{ if }c_{\sigma(i)}=1-b_{\sigma(i)}.
\end{cases}
\end{align}

First, assume that $k^{*}\in\{1,\ldots,n\}$ (i.e., the ``otherwise'' case in \eqn{discretization-function-value} happens).
By \eqn{discrete-line-4-1},
\begin{align}
x^{*}_{\sigma(k^{*})}\in\Big\{\frac{k^{*}}{2n+1},1-\frac{k^{*}}{2n+1}\Big\};\quad x^{*}_{\sigma(i)}\notin\Big\{\frac{k^{*}}{2n+1},1-\frac{k^{*}}{2n+1}\Big\}\quad\forall\,i\neq k^{*},
\end{align}
which implies that for all $i\neq k^{*}$, $|x^{*}_{\sigma(i)}-c_{\sigma(i)}|\neq 1-\frac{k^{*}}{2n+1}$ since $c_{\sigma(i)}\in\{0,1\}$. As a result, we must have
\begin{align}
|x^{*}_{\sigma(k^{*})}-c_{\sigma(k^{*})}|=1-\frac{k^{*}}{2n+1}.
\end{align}
Together with \eqn{discrete-line-4-1}, this implies
\begin{align}\label{eqn:discrete-line-4-2}
c_{\sigma(k^{*})}=1-b_{\sigma(k^{*})}.
\end{align}

For any $i<k^{*}$, if $c_{\sigma(i)}=1-b_{\sigma(i)}$, then \eqn{discrete-line-4-1} implies that
\begin{align}
f(x^{*})\geq |x^{*}_{\sigma(i)}-c_{\sigma(i)}|=1-\frac{i}{2n+1}>1-\frac{k^{*}}{2n+1},
\end{align}
which contradicts with the assumption that $f(x^{*})=1-\frac{k^{*}}{2n+1}$. Therefore, we must have
\begin{align}\label{eqn:discrete-line-4-3}
c_{\sigma(i)}=b_{\sigma(i)}\quad\forall\,i\in\{1,\ldots,k^{*}-1\}.
\end{align}

Recall that the decreasing order of $\{x_{1},\ldots,x_{n},1-x_{1},\ldots,1-x_{n}\}$ is
\begin{align}\label{eqn:discrete-line-4-4}
&b_{\sigma(1)}x_{\sigma(1)}+(1-b_{\sigma(1)})(1-x_{\sigma(1)})\geq\cdots\geq b_{\sigma(n)}x_{\sigma(n)}+(1-b_{\sigma(n)})(1-x_{\sigma(n)}) \nonumber \\
&\quad \geq (1-b_{\sigma(n)})x_{\sigma(n)}+b_{\sigma(n)}(1-x_{\sigma(n)})\geq\cdots\geq (1-b_{\sigma(1)})x_{\sigma(1)}+b_{\sigma(1)}(1-x_{\sigma(1)}).
\end{align}
Based on \eqn{discrete-line-4-2}, \eqn{discrete-line-4-3}, and \eqn{discrete-line-4-4}, we next prove
\begin{align}\label{eqn:discrete-line-4-5}
|x_{\sigma(k^{*})}-c_{\sigma(k^{*})}|\geq |x_{\sigma(i)}-c_{\sigma(i)}|\quad \forall\,i\in\range{n}.
\end{align}
If \eqn{discrete-line-4-5} holds, it implies
\begin{align}\label{eqn:discrete-line-4-6}
f(x)=\max_{i\in\range{n}}|x_{i}-c_{i}|=|x_{\sigma(k^{*})}-c_{\sigma(k^{*})}|.
\end{align}
If $b_{\sigma(k^{*})}=0$, then \eqn{discrete-line-4-2} implies $c_{\sigma(k^{*})}=1$, \eqn{discrete-line-4-6} implies $f(x)=1-x_{\sigma(k^{*})}$, and the output in \lin{discretization-4} satisfies
\begin{align}
b_{\sigma(k^{*})}x_{\sigma(k^{*})}+(1-b_{\sigma(k^{*})})(1-x_{\sigma(k^{*})})=1-x_{\sigma(k^{*})}=f(x);
\end{align}
If $b_{\sigma(k^{*})}=1$, then \eqn{discrete-line-4-2} implies $c_{\sigma(k^{*})}=0$, \eqn{discrete-line-4-6} implies $f(x)=x_{\sigma(k^{*})}$, and the output in \lin{discretization-4} satisfies
\begin{align}
b_{\sigma(k^{*})}x_{\sigma(k^{*})}+(1-b_{\sigma(k^{*})})(1-x_{\sigma(k^{*})})=x_{\sigma(k^{*})}=f(x).
\end{align}
The correctness of \lin{discretization-4} follows.

It remains to prove \eqn{discrete-line-4-5}. We divide its proof into two parts:
\begin{itemize}
\item Suppose $i<k^{*}$. By \eqn{discrete-line-4-4}, we have
\begin{align}\label{eqn:discrete-line-4-7}
b_{\sigma(k^{*})}x_{\sigma(k^{*})}+(1-b_{\sigma(k^{*})})(1-x_{\sigma(k^{*})})
\geq(1-b_{\sigma(i)})x_{\sigma(i)}+b_{\sigma(i)}(1-x_{\sigma(i)}).
\end{align}
\begin{itemize}
\item If $b_{\sigma(k^{*})}=0$ and $b_{\sigma(i)}=0$, we have $c_{\sigma(k^{*})}=1$ and $c_{\sigma(i)}=0$ by \eqn{discrete-line-4-2} and \eqn{discrete-line-4-3}, respectively; \eqn{discrete-line-4-7} reduces to $1-x_{\sigma(k^{*})}\geq x_{\sigma(i)}$;
\item If $b_{\sigma(k^{*})}=0$ and $b_{\sigma(i)}=1$, we have $c_{\sigma(k^{*})}=1$ and $c_{\sigma(i)}=1$ by \eqn{discrete-line-4-2} and \eqn{discrete-line-4-3}, respectively; \eqn{discrete-line-4-7} reduces to $1-x_{\sigma(k^{*})}\geq 1-x_{\sigma(i)}$;
\item If $b_{\sigma(k^{*})}=1$ and $b_{\sigma(i)}=0$, we have $c_{\sigma(k^{*})}=0$ and $c_{\sigma(i)}=0$ by \eqn{discrete-line-4-2} and \eqn{discrete-line-4-3}, respectively; \eqn{discrete-line-4-7} reduces to $x_{\sigma(k^{*})}\geq x_{\sigma(i)}$;
\item If $b_{\sigma(k^{*})}=1$ and $b_{\sigma(i)}=1$, we have $c_{\sigma(k^{*})}=0$ and $c_{\sigma(i)}=1$ by \eqn{discrete-line-4-2} and \eqn{discrete-line-4-3}, respectively; \eqn{discrete-line-4-7} reduces to $x_{\sigma(k^{*})}\geq 1-x_{\sigma(i)}$.
\end{itemize}
In each case, the resulting expression is exactly \eqn{discrete-line-4-5}.
Overall, we see that \eqn{discrete-line-4-5} is always true when $i<k^{*}$.

\item Suppose $i>k^{*}$. By \eqn{discrete-line-4-4}, we have
\begin{align}
b_{\sigma(k^{*})}x_{\sigma(k^{*})}+(1-b_{\sigma(k^{*})})(1-x_{\sigma(k^{*})})
&\geq b_{\sigma(i)}x_{\sigma(i)}+(1-b_{\sigma(i)})(1-x_{\sigma(i)}); \label{eqn:discrete-line-4-8a} \\
b_{\sigma(k^{*})}x_{\sigma(k^{*})}+(1-b_{\sigma(k^{*})})(1-x_{\sigma(k^{*})})
&\geq(1-b_{\sigma(i)})x_{\sigma(i)}+b_{\sigma(i)}(1-x_{\sigma(i)}). \label{eqn:discrete-line-4-8b}
\end{align}
\begin{itemize}
\item If $b_{\sigma(k^{*})}=0$, we have $c_{\sigma(k^{*})}=1$ by \eqn{discrete-line-4-2}; \eqn{discrete-line-4-8a} and \eqn{discrete-line-4-8b} give $1-x_{\sigma(k^{*})}\geq \max\{x_{\sigma(i)},1-x_{\sigma(i)}\}$;
\item If $b_{\sigma(k^{*})}=1$, we have $c_{\sigma(k^{*})}=0$ by \eqn{discrete-line-4-2}; \eqn{discrete-line-4-8a} and \eqn{discrete-line-4-8b} give $x_{\sigma(k^{*})}\geq \max\{x_{\sigma(i)},1-x_{\sigma(i)}\}$.
\end{itemize}
Both cases imply \eqn{discrete-line-4-5}, so we see this also holds for $i>k^{*}$.
\end{itemize}

The same proof works when $k^{*}=n+1$. In this case, there is no $i\in\range{n}$ such that $i>k^{*}$; on the other hand, when $i<k^{*}$, we replace \eqn{discrete-line-4-7} by
\begin{align}
(1-b_{\sigma(n)})x_{\sigma(n)}+b_{\sigma(n)}(1-x_{\sigma(n)})\geq (1-b_{\sigma(i)})x_{\sigma(i)}+b_{\sigma(i)}(1-x_{\sigma(i)}),
\end{align}
and the argument proceeds unchanged.
\end{proof}

\subsection{Optimality of \thm{main-evaluation}}\label{sec:decision-optimality}
In this section, we prove that the lower bound in \thm{main-evaluation} is optimal (up to poly-logarithmic factors in $n$) for the max-norm optimization problem:
\begin{theorem}\label{thm:main-evaluation-optimal}
Let $f_{c}\colon[0,1]^{n}\rightarrow [0,1]$ be an objective function for the max-norm optimization problem (\defn{max-OPT}). Then there exists a quantum algorithm that outputs an $\tilde{x}\in [0,1]^{n}$ satisfying \eqn{max-norm-main} with $\eps=1/3$ using $O(\sqrt{n}\log n)$ quantum queries to $O_{f}$, with success probability at least $0.9$.
\end{theorem}
\noindent
In other words, the quantum query complexity of the max-norm optimization problem is $\tilde{\Theta}(\sqrt{n})$.

We prove \thm{main-evaluation-optimal} also using search with wildcards (\thm{wildcard}).

\begin{proof}
It suffices to prove that one query to the wildcard query model $O_{c}$ in \eqn{wildcard-defn} can be simulated by one query to $O_{f_{c}}$, where the $c$ in \eqn{max-OPT-defn} is the string $c$ in the wildcard query model.

Assume that we query $(T,y)$ using the wildcard query model. Then we query $O_{f_{c}}(x^{(T,y)})$ where for all $i\in\range{n}$,
\begin{align}
x_{i}^{(T,y)}=\begin{cases}
      \frac{1}{2} & \text{if } i\notin T; \\
      0 & \text{if } i\in T\text{ and }y_{i}=0; \\
      1 & \text{if } i\in T\text{ and }y_{i}=1.
    \end{cases}
\end{align}

If $c_{|T}=y$, then
\begin{itemize}
\item if $|T|=n$ (i.e., $T=\range{n}$), then
\begin{align}
f_{c}(x)=\max_{i\in\range{n}}|x_{i}^{(T,y)}-c_{i}|=0
\end{align}
because for any $i\in\range{n}$, $x_{i}^{(T,y)}=y_{i}=c_{i}$;
\item if $|T|\leq n-1$, then
\begin{align}
f_{c}(x)=\max_{i\in\range{n}}|x_{i}^{(T,y)}-c_{i}|+g_{i}=\frac{1}{2},
\end{align}
because for all $i\in T$ we have $x_{i}^{(T,y)}=y_{i}=c_{i}$ and hence $|x_{i}^{(T,y)}-c_{i}|=0$, and for all $i\notin T$ we have $|x_{i}^{(T,y)}-c_{i}|=|\frac{1}{2}-c_{i}|=\frac{1}{2}$.
\end{itemize}
Therefore, if $c_{|T}=y$, then we must have $f_{c}(x^{(T,y)})\in\big\{0,\frac{1}{2}\big\}$.

On the other hand, if $c_{|T}\neq y$, then there exists an $i_{0}\in T$ such that $c_{i_{0}}\neq y_{i_{0}}$. This implies $x_{i_{0}}^{(T,y)}=1-c_{i_{0}}$; as a result, $f_{c}(x^{(T,y)})=1$ because on the one hand $f_{c}(x^{(T,y)})\geq|1-c_{i_{0}}-c_{i_{0}}|=1$, and on the other hand $f_{c}(x^{(T,y)})\leq 1$ as $|x_{i}^{(T,y)}-c_{i}|\leq 1$ for all $i\in\range{n}$.

Notice that the sets $\big\{0,\frac{1}{2}\big\}$ and $\{1\}$ do not intersect. Therefore, after we query $O_{f_{c}}(x^{(T,y)})$ and obtain the output, we can tell $Q_{s}(T,y)=1$ in \eqn{wildcard-defn} if $O_{f_{c}}(x^{(T,y)})\in\big\{0,\frac{1}{2}\big\}$, and output $Q_{s}(T,y)=0$ if $O_{f_{c}}(x^{(T,y)})=1$. In all, this gives a simulation of one query to the wildcard query model $O_{c}$ by one query to $O_{f_{c}}$.

As a result of \thm{wildcard}, there is a quantum algorithm that outputs the $c$ in \eqn{max-OPT-defn} using $O(\sqrt{n}\log n)$ quantum queries to $O_{f}$. If we take $\tilde{x}=c$, then $f_{c}(\tilde{x})=\max_{i}|c_{i}-c_{i}|=0$, which is actually the optimal solution with $\epsilon=0$ in \eqn{max-norm-main}. This establishes \thm{main-evaluation-optimal}.
\end{proof}
\end{appendix}

\end{document}